\documentclass[sts,preprint]{imsart}
\usepackage{geometry} 
\usepackage{graphicx}	
\usepackage{amssymb, amsmath, amsthm}
\usepackage{natbib}
\usepackage{color}
\usepackage[table]{xcolor}
\usepackage{subfigure}
\usepackage{multirow}
\usepackage{verbatim}
\usepackage{enumerate}
\usepackage[colorlinks,citecolor=blue,urlcolor=blue]{hyperref}

\usepackage{tikz}
\usetikzlibrary{intersections, backgrounds}
\definecolor{highlight}{RGB}{180, 31, 180}
\definecolor{gray80}{gray}{0.8}

\newcommand{\dd}{ \mathrm{d} }

\theoremstyle{plain}
\newtheorem{theorem}{Theorem}
\newtheorem{lemma}[theorem]{Lemma}
\newtheorem{corollary}[theorem]{Corollary}

\theoremstyle{definition}
\newtheorem{definition}{Definition}
\newtheorem{example}{Example}
\theoremstyle{remark}

\theoremstyle{plain}

\begin{document}

\begin{frontmatter}

\title{The Geometric Foundations of Hamiltonian Monte Carlo}
\runtitle{The Geometric Foundations of Hamiltonian Monte Carlo}

\begin{aug}
  \author{Michael Betancourt%
  \ead[label=e1]{betanalpha@gmail.com}},
  \author{Simon Byrne},
  \author{Sam Livingstone},
  \and
  \author{Mark Girolami}

  \runauthor{Betancourt et al.}

  \address{Michael Betancourt is a Postdoctoral Research Associate at the University of Warwick, 
                 Coventry CV4 7AL, UK \printead{e1}.  Simon Byrne is an EPSRC Postdoctoral Research
                 Fellow at University College London, Gower Street, London, WC1E 6BT.  Sam
                 Livingstone is a PhD candidate at University College London, Gower Street, London, 
                 WC1E 6BT.  Mark Girolami is an ESPRC Established Career Research Fellow at
                 the University of Warwick, 
                 Coventry CV4 7AL, UK.}

\end{aug}

\begin{abstract}
Although Hamiltonian Monte Carlo has proven an empirical success, the lack of a rigorous 
theoretical understanding of the algorithm has in many ways impeded both principled
developments of the method and use of the algorithm in practice.  In this paper we develop the 
formal foundations of the algorithm through the construction of measures on smooth manifolds, 
and demonstrate how the theory naturally identifies efficient implementations and motivates 
promising generalizations.
\end{abstract}


\end{frontmatter}

The frontier of Bayesian inference requires algorithms capable of fitting complex models
with hundreds, if not thousands of parameters, intricately bound together with nonlinear 
and often hierarchical correlations.  Hamiltonian Monte Carlo~\citep{DuaneEtAl:1987, Neal:2011} 
has proven tremendously successful at extracting inferences from these models, with 
applications spanning computer 
science~\citep{SutherlandEtAl:2013, TangEtAl:2013}, 
ecology~\citep{SchofieldEtAl:2013, TeradaEtAl:2013}, 
epidemiology~\citep{CowlingEtAl:2012}, 
linguistics~\citep{HusainEtAl:2014}, pharmacokinetics~\citep{WeberEtAl:2014},
physics~\citep{JascheEtAl:2010, PorterEtAl:2014, SandersEtAl:2014, WangEtAl:2014},
and political science~\citep{GhitzaEtAl:2013}, to name a few.  Despite such widespread empirical
success, however, there remains an air of mystery concerning the efficacy of the algorithm.  

This lack of understanding not only limits the adoption of Hamiltonian Monte Carlo but may also foster 
unprincipled and, ultimately, fragile implementations that restrict the scalability of the algorithm.  
Consider, for example, the Compressible Generalized Hybrid Monte Carlo scheme of 
\cite{FangEtAl:2014} and the particular implementation in Lagrangian Dynamical Monte 
Carlo~\citep{LanEtAl:2012}.  In an effort to reduce the computational burden of the algorithm,
the authors sacrifice the costly volume-preserving numerical integrators typical to Hamiltonian Monte
Carlo.  Although this leads to improved performance in some low-dimensional models, the performance
rapidly diminishes with increasing model dimension~\citep{LanEtAl:2012} in sharp contrast to standard 
Hamiltonian Monte Carlo.  Clearly, the volume-preserving numerical integrator is somehow critical to 
scalable performance; but why?

In this paper we develop the theoretical foundation of Hamiltonian Monte Carlo in order to answer 
questions like these.  We demonstrate how a formal understanding naturally identifies the properties 
critical to the success of the algorithm, hence immediately providing a framework for robust implementations.  
Moreover, we discuss how the theory motivates several generalizations that may extend the success of 
Hamiltonian Monte Carlo to an even broader array of applications.  

We begin by considering the properties of efficient Markov kernels and possible strategies for 
constructing those kernels.  This construction motivates the use of tools in differential geometry, and 
we continue by curating a coherent theory of probabilistic measures on smooth manifolds.  In the 
penultimate section we show how that theory provides a skeleton for the development, implementation, 
and formal analysis of Hamiltonian Monte Carlo.  Finally, we discuss how this formal perspective directs 
generalizations of the algorithm.

Without a familiarity with differential geometry a complete understanding of this work will be a challenge, 
and we recommend that readers without a background in the subject only scan through Section
\ref{sec:measures_on_manifolds} to develop some intuition for the probabilistic interpretation of forms, 
fiber bundles, Riemannian metrics, and symplectic forms, as well as the utility of Hamiltonian flows.  For 
those readers interesting in developing new implementations of Hamiltonian Monte Carlo we recommend 
a more careful reading of these sections and suggest introductory literature on the mathematics 
necessary to do so in the introduction of Section \ref{sec:measures_on_manifolds}.


\section{Constructing Efficient Markov Kernels}

Bayesian inference is conceptually straightforward: the information about a system is first modeled 
with the construction of a posterior distribution, and then statistical questions can be answered by 
computing expectations with respect to that distribution.  Many of the limitations of Bayesian 
inference arise not in the modeling of a posterior distribution but rather in computing the subsequent 
expectations.  Because it provides a generic means of estimating these expectations, Markov Chain 
Monte Carlo has been critical to the success of the Bayesian methodology in practice.

In this section we first review the Markov kernels intrinsic to Markov Chain Monte Carlo
and then consider the dynamic systems perspective to motivate a strategy for
constructing Markov kernels that yield computationally efficient inferences.


\subsection{Markov Kernels}

Consider a probability space,
\begin{equation*}
( Q, \mathcal{B} \! \left( Q \right) \!, \varpi ),
\end{equation*} 
with an $n$-dimensional sample space, $Q$, the Borel $\sigma$-algebra over $Q$, 
$\mathcal{B} \! \left( Q \right)$, and a distinguished probability measure, $\varpi$.  In a 
Bayesian application, for example, the distinguished measure would be the posterior 
distribution and our ultimate goal would be the estimation of expectations with respect
to the posterior, $\mathbb{E}_{\varpi} \! \left[ f \right]$.

A \textit{Markov kernel}$, \tau$, is a map from an element of the sample space and 
the $\sigma$-algebra to a probability,
\begin{equation*}
\tau : Q \times \mathcal{B} \! \left( Q \right) \rightarrow \left[0, 1\right],
\end{equation*}
such that the kernel is a measurable function in the first argument,
\begin{equation*}
\tau \! \left( \cdot, A \right) : Q \rightarrow Q, \, \forall A \in \mathcal{B} \! \left( Q \right),
\end{equation*}
and a probability measure in the second argument,
\begin{equation*}
\tau \! \left( q, \cdot \right) : \mathcal{B} \! \left( Q \right) \rightarrow \left[0, 1 \right], \, \forall q \in Q.
\end{equation*}
By construction the kernel defines a map,
\begin{equation*}
\tau : Q \rightarrow  \mathcal{P} \! \left( Q \right),
\end{equation*}
where $\mathcal{P} \! \left( Q \right)$ is the space of probability measures over $Q$;
intuitively, at each point in the sample space the kernel defines a probability measure 
describing how to sample a new point.

By averaging the Markov kernel over all initial points in the state space we can construct a 
\textit{Markov transition} from probability measures to probability measures,
\begin{equation*}
\mathcal{T} : \mathcal{P} \! \left( Q \right) \rightarrow \mathcal{P} \! \left( Q \right),
\end{equation*}
by
\begin{align*}
\varpi' \! \left( A \right) 
= \varpi \mathcal{T} \! \left( A \right)
= \int \tau \! \left( q, A \right) \varpi \! \left( \mathrm{d} q \right), 
\, \forall q \in Q, \, A \in \mathcal{B} \! \left( Q \right).
\end{align*}
When the transition is aperiodic, irreducible, Harris recurrent, and preserves the target measure, 
$\varpi \mathcal{T} = \varpi$, its repeated application generates a \textit{Markov chain} that will
eventually explore the entirety of $\varpi$.  Correlated samples, $\left( q_{0}, q_{1}, \ldots, q_{N} \right)$ 
from the Markov chain yield 
\textit{Markov Chain Monte Carlo estimators} of any expectation~\citep{RobertsEtAl:2004, MeynEtAl:2009}.
Formally, for any integrable function $f \in L^{1} \! \left( Q, \varpi \right)$ we can construct estimators,
\begin{equation*}
\hat{f}_{N} \! \left( q_{0} \right) = 
\frac{1}{N} \sum_{n = 0}^{N} f \! \left( q_{n} \right),
\end{equation*}
that are asymptotically consistent for any initial $q_{0} \in Q$,
\begin{equation*}
\lim_{N \rightarrow \infty} \hat{f}_{N} \! \left( q_{0} \right)
\xrightarrow{\mathcal{P}}
\mathbb{E}_{\varpi} \! \left[ f \right].
\end{equation*}
Here $\delta_{q}$ is the \textit{Dirac measure} that concentrates on $q$,
\begin{equation*}
\delta_{q} \! \left( A \right) \propto
\left\{
\begin{array}{rr}
0, & q \notin A\\
1, & q \in A
\end{array} 
\right. , \, q \in Q, A \in \mathcal{B} \! \left( Q \right).
\end{equation*}

In practice we are interested not just in Markov chains that explore the target distribution as
$N \rightarrow \infty$ but in Markov chains that can explore and yield precise Markov Chain Monte
Carlo estimators in only a finite number of transitions.  From this perspective the efficiency of a Markov 
chain can be quantified in terms of the \textit{autocorrelation}, which measures the dependence of any 
square integrable test function, $f \in L^{2} \! \left( Q, \varpi \right)$, before and after the application of 
the Markov transition
\begin{align*}
\rho \! \left[ f \right] 
&\equiv
\frac{
\int f \! \left( q_{1} \right) f \! \left( q_{2} \right) 
\tau \! \left( q_{1}, \mathrm{d}q_{2} \right) \varpi \! \left( \mathrm{d} q_{1} \right)
- \int f \! \left( q_{2} \right) \varpi \! \left( \mathrm{d} q_{2} \right) 
\int f \! \left( q_{1} \right) \varpi \! \left( \mathrm{d} q_{1} \right)
}{
\int f^{2} \! \left( q \right) \varpi \! \left( \mathrm{d} q \right)
- \left( \int f \! \left( q \right) \varpi \! \left( \mathrm{d} q \right) \right)^{2}
}.
\end{align*}
In the best case the Markov kernel reproduces the target measure,
\begin{equation*}
\tau \! \left( q, A \right) = \varpi \! \left( A \right), \, \forall q \in Q,
\end{equation*}
and the autocorrelation vanishes for all test functions, $\rho \! \left[ f \right] = 0$. 
Alternatively, a Markov kernel restricted to a Dirac measure at the initial point,
\begin{equation*}
\tau \! \left( q, A \right) = \delta_{q} \! \left( A \right),
\end{equation*}
moves nowhere and the autocorrelations saturate for any test function,
$\rho \! \left[ f \right] = 1$.  Note that we are disregarding anti-autocorrelated chains, 
whose performance is highly sensitive to the particular $f$ under consideration.

Given a target measure, any Markov kernel will lie in between these two extremes; 
the more of the target measure a kernel explores the smaller the autocorrelations, 
while the more localized the exploration to the initial point the larger the autocorrelations.
Unfortunately, common Markov kernels like 
\textit{Gaussian Random walk Metropolis}~\citep{RobertEtAl:1999} and the 
\textit{Gibbs sampler}~\citep{GemanEtAl:1984, GelfandEtAl:1990} degenerate into 
local exploration, and poor efficiency, when targeting the complex distributions of
interest.  Even in two-dimensions, for example, nonlinear correlations in the target  
distribution constrain the $n$-step transition kernels to small neighborhoods around
the initial point (Figure \ref{fig:poor_metro_and_gibbs}).

\begin{figure}
\centering
\subfigure[]{ \includegraphics[width=2.55in]{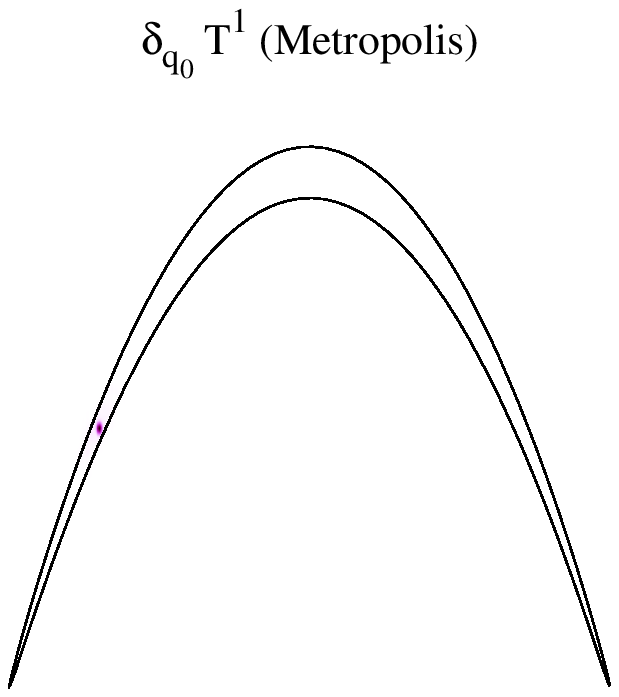} }
\subfigure[]{ \includegraphics[width=2.55in]{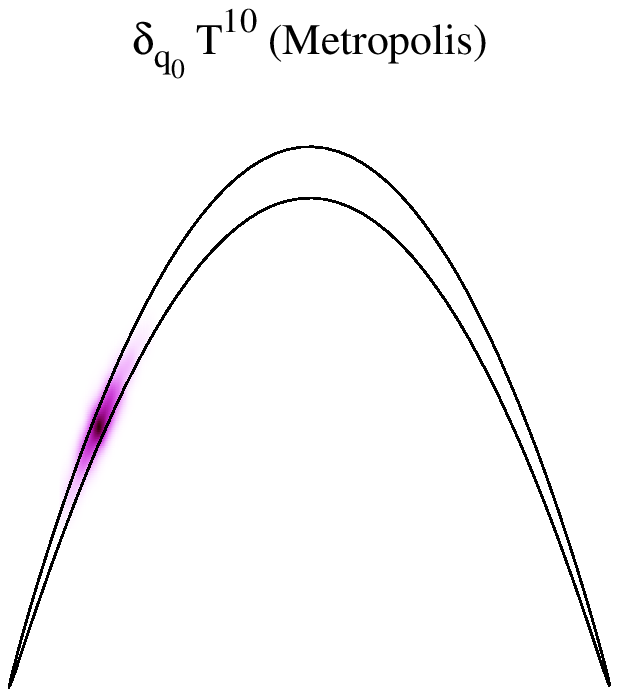} }
\subfigure[]{ \includegraphics[width=2.55in]{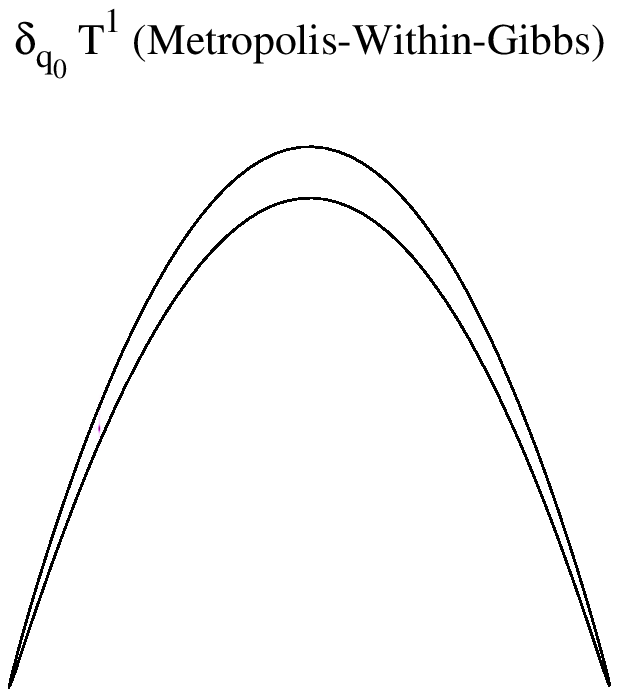} }
\subfigure[]{ \includegraphics[width=2.55in]{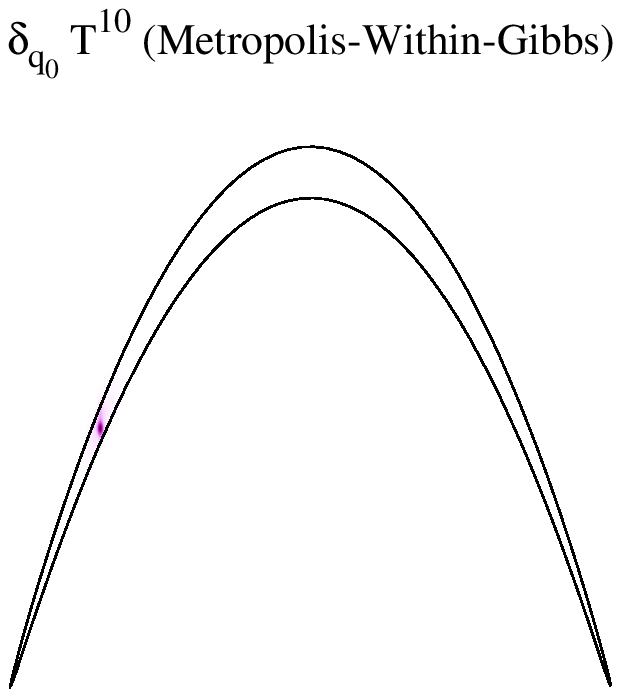} }
\caption{Both (a, b) Random Walk Metropolis and (c, d) the Gibbs sampler are stymied by
complex distributions, for example a warped Gaussian distribution~\citep{HaarioEtAl:2001}
on the sample space $Q = \mathbb{R}^{2}$, here represented with a 95\% probability contour.  
Even when optimally tuned~\citep{RobertsEtAl:1997}, both Random Walk Metropolis and Random 
Walk Metropolis-within-Gibbs kernels concentrate around the initial point, even after multiple 
iterations.}
\label{fig:poor_metro_and_gibbs}
\end{figure}

In order for Markov Chain Monte Carlo to perform well on these contemporary problems
we need to be able to engineer Markov kernels that maintain exploration, and hence small 
autocorrelations, when targeting intricate distributions.  The construction of such kernels is 
greatly eased with the use of measure-preserving maps.


\subsection{Markov Kernels Induced From Measure-Preserving Maps}

Directly constructing a Markov kernel that targets $\varpi$, let alone an efficient Markov kernel,
can be difficult.  Instead of constructing a kernel directly, however, we can construct
one indirectly be defining a family of measure-preserving maps~\citep{Petersen:1989}.

Formally, let $\Gamma$ be some space of continuous, bijective maps, or isomorphisms, from the 
space into itself,
\begin{equation*}
t : Q \rightarrow Q, \, \forall t \in \Gamma,
\end{equation*}
that each preserves the target measure,
\begin{equation*}
t_{*} \varpi = \varpi,
\end{equation*}
where the pushforward measure, $t_{*} \varpi$, is defined as
\begin{equation*}
\left( t_{*} \varpi \right) \! \left( A \right) \equiv \left( \varpi \circ t^{-1} \right) \! \left( A \right), \,
\forall A \in \mathcal{B} \! \left( Q \right).
\end{equation*}
If we can define a $\sigma$-algebra, $\mathcal{G}$, over this space then the choice of a 
distinguished measure over $\mathcal{G}$, $\gamma$, defines a probability space,
\begin{equation*} 
( \Gamma, \mathcal{G}, \gamma ),
\end{equation*}
which induces a Markov kernel by
\begin{equation} \label{eqn:kernel_from_isomorphisms}
\tau \! \left( q, A\right) \equiv
\int_{\Gamma} \gamma \! \left( \dd t \right) \mathbb{I}_{A} \! \left( t \left( q \right) \right),
\end{equation}
where $\mathbb{I}$ is the indicator function,
\begin{equation*}
\mathbb{I}_{A} \! \left( q \right) \propto
\left\{
\begin{array}{rr}
0, & q \notin A\\
1, & q \in A
\end{array} 
\right. , \, q \in Q, A \in \mathcal{B} \! \left( Q \right).
\end{equation*}
In other words, the kernel assigns a probability to a set, $A \in \mathcal{B} \! \left( Q \right)$,
by computing the measure of the preimage of that set, $t^{-1} \! \left( A \right)$, averaged 
over all isomorphisms in $\Gamma$.  Because each $t$ preserves the target measure, 
so too will their convolution and, consequently, the Markov transition induced by the kernel.

This construction provides a new perspective on the limited performance of existing
algorithms.
\begin{example}
We can consider Gaussian Random Walk Metropolis, for example, as being generated by
random, independent translations of each point in the sample space,
\begin{align*}
  t_{\epsilon, \eta} :& \, q \mapsto 
           q + \epsilon \, 
           \mathbb{I} \! \left( \eta < \frac{ f \! \left( q + \epsilon \right) }{ f \! \left( q \right) } \right) \\
  & \, \epsilon \sim \mathcal{N} \! \left( 0, \Sigma \right) \\
  & \, \eta \sim U \! \left[ 0, 1 \right],
\end{align*}
where $f$ is the density of $\varpi$ with respect to the Lebesgue measure on $\mathbb{R}^{n}$.
When targeting complex distributions either $\epsilon$ or the support of the indicator will be 
small and the resulting translations barely perturb the initial state.
\end{example}
\begin{example}
The random scan Gibbs sampler is induced by axis-aligned translations,
\begin{align*}
  t_{i, \eta} :& \, q_{i} \rightarrow P_{i}^{-1} \! \left( \eta \right) \\
  & \, i \sim U \! \left\{ 1, \ldots, n \right\} \\
  & \, \eta \sim U \! \left[ 0, 1 \right],
\end{align*}
where
\begin{equation*}
P_{i} \! \left( q_{i} \right) = \int_{\infty}^{q_{i}} \varpi \! \left( \dd \tilde{q}_{i} | q \right)
\end{equation*}
is the cumulative distribution function of the $i$th conditional measure.  When the target
distribution is strongly correlated, the conditional measures concentrate near the initial $q$
and, as above, the translations are stunted.
\end{example}

In order to define a Markov kernel that remains efficient in difficult problems we need
measure-preserving maps whose domains are not limited to local exploration.  Realizations 
of Langevin diffusions~\citep{Oksendal:2003}, for example, yield measure-preserving maps 
that diffuse across the entire target distribution.  Unfortunately that diffusion tends to expand
across the target measures only slowly (Figure \ref{fig:poor_langevin_realization}): for any 
finite diffusion time the resulting Langevin kernels are localized around the initial point 
(Figure \ref{fig:poor_langevin}).  What we need are more coherent maps that avoid such diffusive 
behavior.

\begin{figure}
\centering
\includegraphics[width=3in]{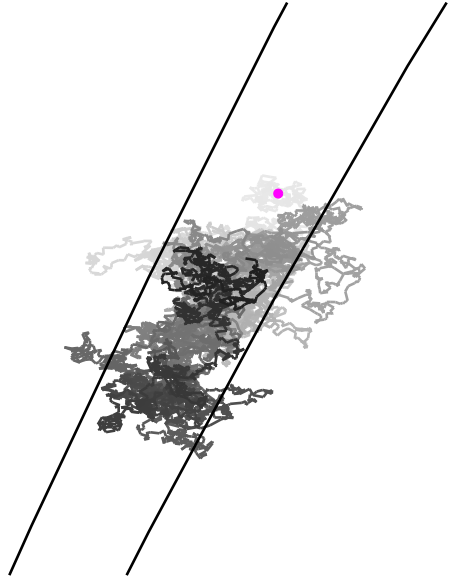}
\caption{Langevin trajectories are, by construction, diffusive, and are just as likely to double back
as they are to move forward.  Consequently even as the diffusion time grows, here to $t=1000$
as the trajectory darkens, realizations of a Langevin diffusion targeting the twisted Gaussian
distribution~\citep{HaarioEtAl:2001} only slowly wander away from the initial point.}
\label{fig:poor_langevin_realization}
\end{figure}

\begin{figure}
\centering
\subfigure[]{ \includegraphics[width=1.75in]{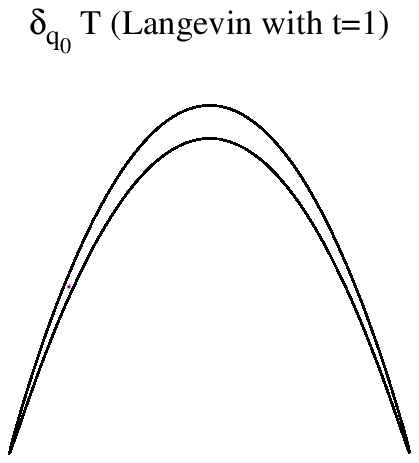} }
\subfigure[]{ \includegraphics[width=1.75in]{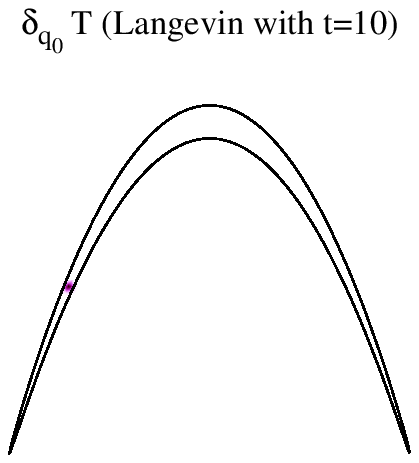} }
\subfigure[]{ \includegraphics[width=1.75in]{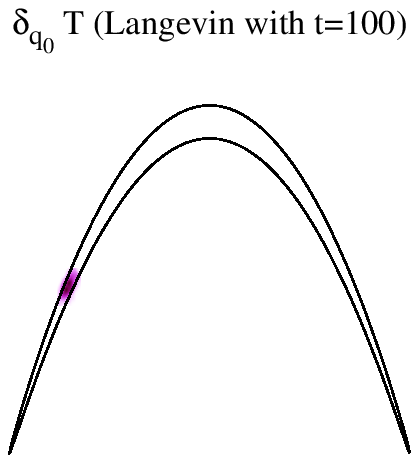} }
\caption{Because of the diffusive nature of the underlying maps, Langevin
kernels expand very slowly with increasing diffusion time, $t$.  For any
reasonable diffusion time the resulting kernels will concentrate around the
initial point, as seen here for a Langevin diffusion targeting the twisted
Gaussian distribution~\citep{HaarioEtAl:2001}.}
\label{fig:poor_langevin}
\end{figure}

One potential candidate for coherent maps are \textit{flows}.  A flow, $\left\{ \phi_{t} \right\}$,
is a family of isomorphisms parameterized by a time, $t$,
\begin{equation*}
\phi_{t} : Q \rightarrow Q, \, \forall t \in \mathbb{R},
\end{equation*}
that form a one-dimensional Lie group on composition,
\begin{align*}
\phi_{t} \circ \phi_{s} &= \phi_{s + t} \\
\phi_{t}^{-1} &= \phi_{-t} \\
\phi_{0} &= \mathrm{Id}_{Q},
\end{align*}
where $\mathrm{Id}_{Q}$ is the natural identity map on $Q$.
Because the inverse of a map is given only by negating $t$, as the time is increased the
resulting $\phi_{t}$ pushes points away from their initial positions and avoids localized
exploration (Figure \ref{fig:diffusion_vs_flow}).  Our final obstacle is in engineering a flow 
comprised of measure-preserving maps.

\begin{figure}
\centering
\includegraphics[width=3in]{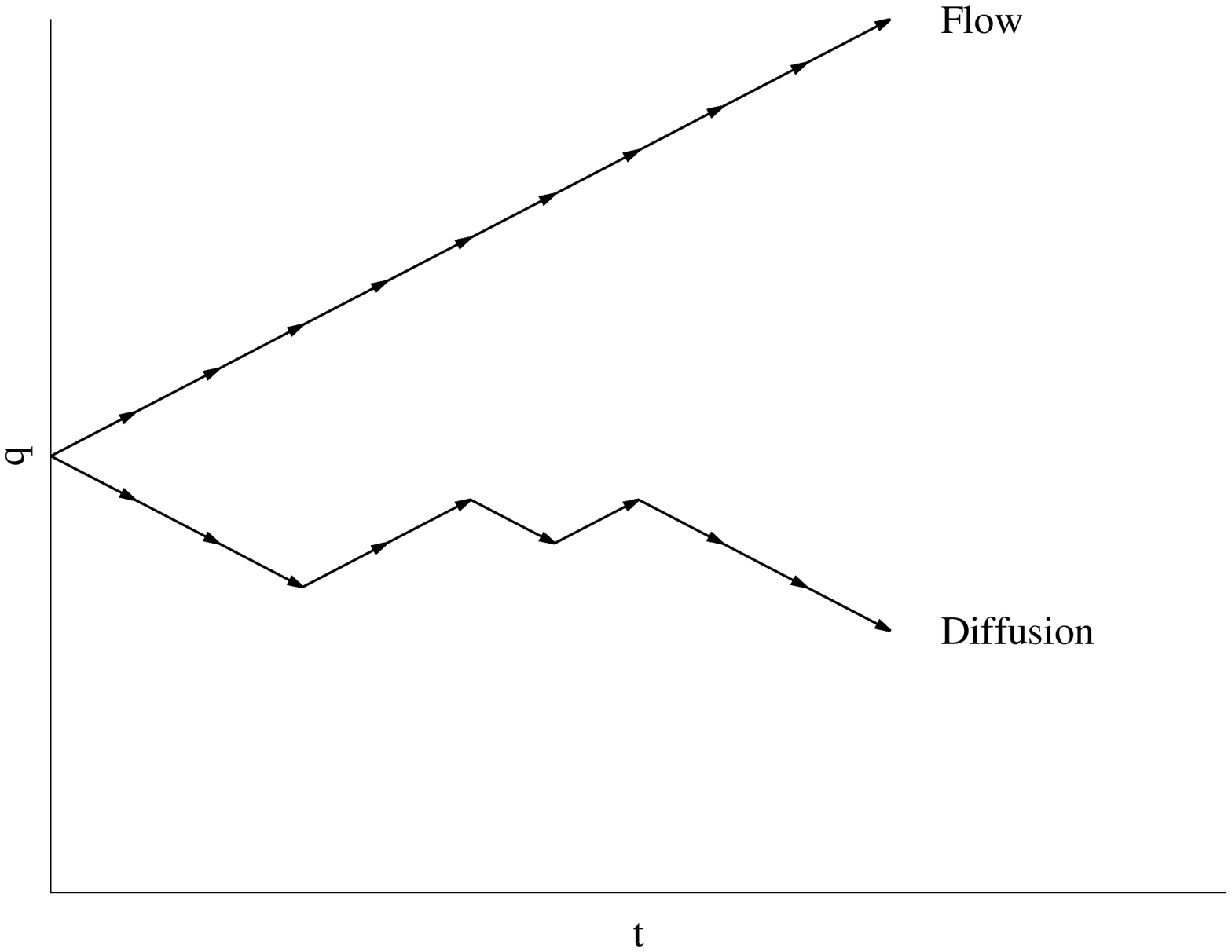}
\caption{Because of the underlying group structure flows cannot double back on themselves
like diffusions, forcing a coherent exploration of the target space.}
\label{fig:diffusion_vs_flow}
\end{figure}

Flows are particularly natural on the \textit{smooth manifolds} of differential geometry, and 
flows that preserve a given target measure can be engineered on one exceptional class of 
smooth manifolds known as \textit{symplectic manifolds}.  If we can understand these manifolds 
probabilistically then we can take advantage of their properties to build Markov kernels with small 
autocorrelations for even the complex, high-dimensional target distributions of practical interest.


\section{Measures on Manifolds} \label{sec:measures_on_manifolds}

In this section we review probability measures on smooth manifolds of increasing sophistication, 
culminating in the construction of measure-preserving flows.

Although we will relate each result to probabilistic theory and introduce intuition where we can,
the formal details in the following require a working knowledge of differential geometry up 
to \cite{Lee:2013}.  We will also use the notation therein throughout the paper.  For readers 
new to the subject but interested in learning more, we recommend the introduction 
in \cite{BaezEtAl:1994}, the applications in~\cite{Schutz:1980, JoseEtAl:1998}, and then 
finally \cite{Lee:2013}.  The theory of symplectic geometry in which we will be particularly 
interested is reviewed in \cite{Schutz:1980, JoseEtAl:1998,  Lee:2013}, with \cite{CannasDaSilva:2001} 
providing the most modern and thorough coverage of the subject.

Smooth manifolds generalize the Euclidean space of real numbers and the corresponding
calculus; in particular, a smooth manifold need only look \textit{locally} like a Euclidean space
(Figures \ref{fig:cylinder}).  This more general space includes Lie groups, Stiefel manifolds, and 
other spaces becoming common in contemporary applications~\citep{ByrneEtAl:2013}, not to 
mention regular Euclidean space as a special case.  It does not, however, include any manifold 
with a discrete topology such as tree spaces.

\begin{figure*}
\centering
\subfigure[]{
\begin{tikzpicture}[scale=0.25, thick]
  
  \draw[] (-5, -8) -- (-5, 8);
  \draw[] (5, -8) -- (5, 8);
  \draw[style=dashed] (-5, -8) arc (180:360:5 and 1);
  \draw[style=dashed] (0, 8) ellipse (5 and 1);
  
  \node at (0, -11) { $Q = \mathbb{S}^{1} \times \mathbb{R}$};
  
\end{tikzpicture}
}
\subfigure[]{
\begin{tikzpicture}[scale=0.25, thick]
  
  \fill[gray80] (-3.7, -8.7) -- +(0, 16.65) arc (220:275:5 and 1) -- +(0, -16.65) arc (275:220:5 and 1);
  \fill[gray80] (1.35, 9.65) -- +(0, -2) arc(-80:-50:5 and 1) -- +(0, 1.5) arc(50:80:5 and 1);
  
  \draw[] (-5, -8) -- (-5, 8);
  \draw[] (0.5, -9.1) -- (0.5, 6.9);
  \draw[] (3.7, 7.9) -- (3.7, 8.5);
  \draw[style=dashed] (3.75, 8.5) arc (40:275:5 and 1);
  \draw[style=dashed] (-5, -8) arc (180:275:5 and 1);
  
  \node at (-5, 10) { $\mathcal{U}_{1} \subset Q$};

  \draw[] (5, -7.5) -- (5, 8.5);
  \draw[] (-3.7, 8) -- +(0, -0.85);
  \draw[] (1.35, 9.65) -- +(0, -0.9);
  \draw[style=dashed] (-3.75, 8) arc (-140:80:5 and 1);
  \draw[style=dashed] (5, -7.5) arc (0:-85:5 and 1);
  
  \node at (5, -11) { $\mathcal{U}_{2} \subset Q$};
  
\end{tikzpicture}
}
\subfigure[]{
\begin{tikzpicture}[scale=0.25, thick]
  
  \draw[] (-10, 1) -- (-10, 9);
  \draw[] (10, 1) -- (10, 9);
  \draw[style=dashed] (-10, 1) -- (10, 1);
  \draw[style=dashed] (-10, 9) -- (10, 9);
  
  \node at (0, 11) { $\psi_{1} \! \left( \mathcal{U}_{1} \right) \subset \mathbb{R}^{2}$};
  
  \fill[gray80] (-10, 1) rectangle (-8, 9);
  \fill[gray80] (10, 1) rectangle (8, 9);
  
  \draw[] (-10, -1) -- (-10, -9);
  \draw[] (10, -1) -- (10, -9);
  \draw[style=dashed] (-10, -1) -- (10, -1);
  \draw[style=dashed] (-10, -9) -- (10, -9);

  \node at (0, -11) { $\psi_{2} \! \left( \mathcal{U}_{2} \right) \subset \mathbb{R}^{2}$};

  \fill[gray80] (-10, -1) rectangle (-8, -9);
  \fill[gray80] (10, -1) rectangle (8, -9);
  
\end{tikzpicture}
}
\caption{(a) The cylinder, $Q = \mathbb{S}^{1} \times \mathbb{R}$, is a nontrivial example 
of a manifold.  Although not globally equivalent to a Euclidean space, (b) the cylinder can 
be covered in two neighborhoods (c) that are themselves isomorphic to an open neighborhood 
in $\mathbb{R}^{2}$.  The manifold becomes smooth when 
$\psi_{1} \circ \psi_{2}^{-1} : \mathbb{R}^{2} \rightarrow \mathbb{R}^{2}$ is a smooth function 
wherever the two neighborhoods intersect (the intersections here shown in gray).  }
\label{fig:cylinder}
\end{figure*}

Formally, we assume that our sample space, $Q$, satisfies the properties of a smooth, connected, and 
orientable $n$-dimensional manifold.  Specifically we require that $Q$ be a Hausdorff and second-countable 
topological space that is locally homeomorphic to $\mathbb{R}^{n}$ and equipped with a differential structure,
\begin{equation*}
\left\{ \mathcal{U}_{\alpha}, \psi_{\alpha} \right\}_{\alpha \in I},
\end{equation*}
consisting of open neighborhoods in Q,
\begin{equation*}
\mathcal{U}_{\alpha} \subset Q,
\end{equation*}
and homeomorphic charts,
\begin{equation*}
\psi_{\alpha} : \mathcal{U}_{\alpha} \rightarrow \mathcal{V}_{\alpha} \subset \mathbb{R}^{n},
\end{equation*}
that are smooth functions whenever their domains overlap (Figure \ref{fig:cylinder}),
\begin{equation*}
\psi_{\beta} \circ \psi_{\alpha}^{-1} \in C^{\infty} \! \left( \mathbb{R}^{n} \right),
\forall \alpha, \beta \, | \, \mathcal{U}_{\alpha} \cap \mathcal{U}_{\beta} \neq \emptyset.
\end{equation*}
Coordinates subordinate to a chart,
\begin{align*}
q^{i} : \,\,& \mathcal{U}_{\alpha} \rightarrow \mathbb{R}
\\
& q \rightarrow \pi_{i} \circ \psi_{\alpha},
\end{align*}
where $\pi_{i}$ is the $i$th Euclidean projection on the image of $\psi_{\alpha}$, provide
local parameterizations of the manifold convenient for explicit calculations.

This differential structure allows us to define calculus on manifolds by applying concepts
from real analysis in each chart.  The differential properties of a function 
$f :  Q \rightarrow \mathbb{R}$, for example, can be studied by considering the entirely-real 
functions,
\begin{equation*}
f \circ \psi_{\alpha}^{-1} : \mathbb{R}^{n} \rightarrow \mathbb{R};
\end{equation*}
because the charts are smooth in their overlap, these local properties define a consistent 
global definition of smoothness.

Ultimately these properties manifest as geometric objects on $Q$, most importantly vector 
fields and differential $k$-forms.  Informally, vector fields specify directions and magnitudes
at each point in the manifold while $k$-forms define multilinear, antisymmetric maps of $k$ such 
vectors to $\mathbb{R}$.  If we consider $n$ linearly-independent vector fields as defining infinitesimal 
parallelepipeds at every point in space, then the action of $n$-forms provides a local sense of
volume and, consequently, integration.  In particular, when the manifold is orientable we can 
define $n$-forms that are everywhere positive and a geometric notion of a measure.

Here we consider the probabilistic interpretation of these volume forms, first on smooth 
manifolds in general and then on smooth manifolds with additional structure: 
fiber bundles, Riemannian manifolds, and symplectic manifolds.  Symplectic manifolds
will be particularly important as they naturally provide measure-preserving flows.  Proofs 
of intermediate lemmas are presented in Appendix \ref{sec:proofs}.


\subsection{Smooth Measures on Generic Smooth Manifolds}

Formally, volume forms are defined as positive, top-rank differential forms,
\begin{equation*}
\mathcal{M} \! \left( Q \right) \equiv
\left\{ \mu \in \Omega^{n} \! \left( Q \right) \, | \, \mu_{q} > 0, \forall q \in Q \right\},
\end{equation*}
where $\Omega^{n} \! \left( Q \right)$ is the space of $n$-forms on $Q$.  By leveraging
the local equivalence to Euclidean space, we can show that these volume forms satisfy 
all of the properties of $\sigma$-finite measures on $Q$ (Figure \ref{fig:cylinder_density}).

\begin{figure*}
\centering
\subfigure[]{
\begin{tikzpicture}[scale=0.25, thick]

  \node[] at (0,0) {\includegraphics[width=2.5cm]{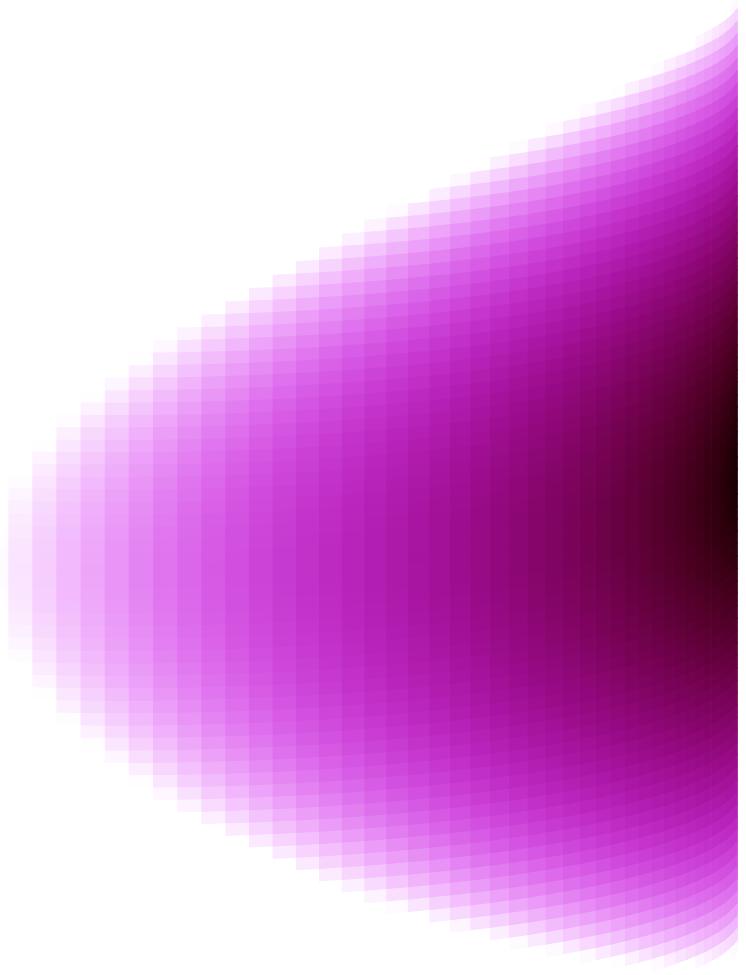}};

  \draw[] (5, -7.5) -- (5, 8.5);
  \draw[] (-3.7, 8) -- +(0, -16);
  \draw[] (1.35, 9.65) -- +(0, -2);
  \draw[style=dashed] (-3.75, 8) arc (-140:80:5 and 1);
  \draw[style=dashed] (5, -7.5) arc (0:-140:5 and 1);
  
  \node at (0, -11) { $\mathcal{U}_{2} \subset Q$};
  
\end{tikzpicture}
}
\subfigure[]{
\begin{tikzpicture}[scale=0.25, thick]
  
  \node[] at (0,0) {\includegraphics[width=5cm]{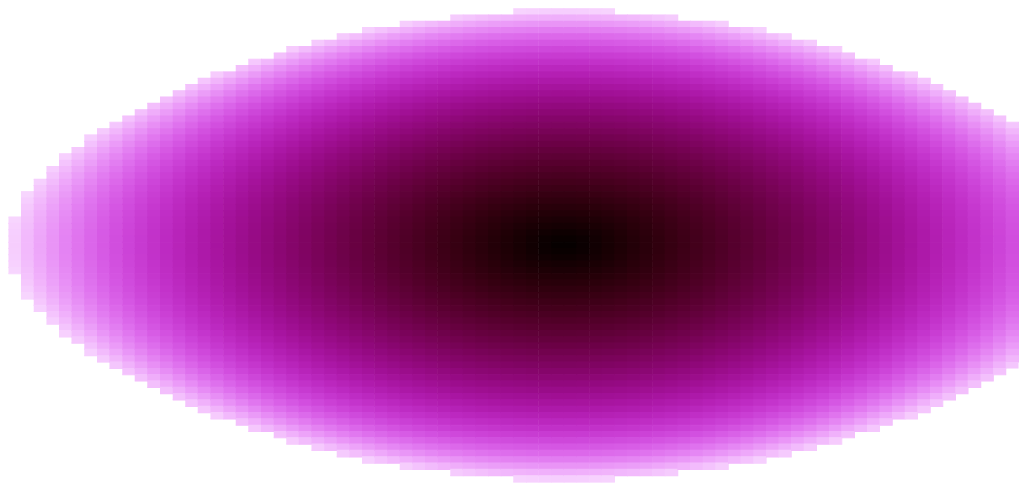}};
  
  \draw[] (-10, -8) -- (-10, 8);
  \draw[] (10, -8) -- (10, 8);
  \draw[style=dashed] (-10, -8) -- (10, -8);
  \draw[style=dashed] (-10, 8) -- (10, 8);
  
  \node at (0, -11) { $\psi_{2} \! \left( \mathcal{U}_{2} \right) \subset \mathbb{R}^{2}$};
  
\end{tikzpicture}
}
\caption{(a) In the neighborhood of a chart, any top-rank differential form is specified 
by its density, $\mu \! \left( q^{1}, \ldots, q^{n} \right)$, with respect to the coordinate volume, 
$\mu = \mu \! \left( q^{1}, \ldots, q^{n} \right) \dd q^{1} \wedge \ldots \wedge \dd q^{n}$,
(b) which pushes forward to a density with respect to the Lebesgue measure in the
image of the corresponding chart.  By smoothly patching together these equivalences,
Lemma \ref{lem:forms_as_measures} demonstrates that these forms are in fact measures. }
\label{fig:cylinder_density}
\end{figure*}

\begin{lemma}
\label{lem:forms_as_measures}

If $Q$ is a positively-oriented, smooth manifold then $\mathcal{M} \! \left( Q \right)$ 
is non-empty and its elements are $\sigma$-finite measures on $Q$.

\end{lemma}

\noindent We will refer to elements of $\mathcal{M} \! \left( Q \right)$ as \textit{smooth
measures} on $Q$.

Because of the local compactness of $Q$, the elements of $\mathcal{M} \! \left( Q \right)$
are not just measures but also Radon measures.  As expected from the Riesz Representation 
Theorem \citep{Folland1999}, any such element also serves as a linear functional via the usual 
geometric definition of integration,
\begin{align*}
\mu : \,
& L^{1} \! \left( Q, \mu \right) \rightarrow \mathbb{R} \\
& f \mapsto \int_{Q} f \mu.
\end{align*}
Consequently, $\left( Q, \mathcal{M} \! \left( Q \right) \right)$ is also a Radon space, which
guarantees the existence of various probabilistic objects such as disintegrations as discussed
below.

Ultimately we are not interested in the whole of $\mathcal{M} \! \left( Q \right)$ 
but rather $\mathcal{P} \! \left( Q \right)$, the subset of volume forms with unit integral,
\begin{equation*}
\mathcal{P} \! \left( Q \right) 
= \left\{ \varpi \in \mathcal{M} \! \left( Q \right) \left| \int_{Q} \varpi = 1 \right. \right\},
\end{equation*}
which serve as probability measures.  Because we can always normalize measures,
$\mathcal{P} \! \left( Q \right)$ is equivalent the finite elements of $\mathcal{M} \! \left( Q \right)$,
\begin{equation*}
\widetilde{\mathcal{M}} \! \left( Q \right) 
= \left\{ \varpi \in \mathcal{M} \! \left( Q \right) \left| \int_{Q} \varpi < \infty \right. \right\},
\end{equation*}
modulo their normalizations.

\begin{corollary}
\label{cor:forms_as_finite_measures}

If $Q$ is a positively-oriented, smooth manifold then $\widetilde{\mathcal{M}} \! \left( Q \right)$, 
and hence $\mathcal{P} \! \left( Q \right)$, is non-empty.

\end{corollary}

\begin{proof}

Because the manifold is paracompact, the prototypical measure constructed in Lemma
\ref{lem:forms_as_measures} can always be chosen such that the measure of the entire
manifold is finite.

\end{proof}


\subsection{Smooth Measures on Fiber Bundles} \label{subsec:fiber_bundles}

Although conditional probability measures are ubiquitous in statistical methodology,
they are notoriously subtle objects to rigorously construct in theory~\citep{Halmos:1950}.  
Formally, a conditional probability measure appeals to a measurable function between 
two generic spaces, $F : R \rightarrow S$, to define measures on $R$ for some subsets 
of $S$ along with an abundance of technicalities.  It is only when $S$ is endowed with 
the quotient topology relative to $F$~\citep{Folland1999, Lee:2011} that we can define 
\textit{regular conditional probability measures} that shed many of the technicalities and 
align with common intuition.  In practice, regular conditional probability measures are 
most conveniently constructed as
\textit{distintegrations}~\citep{ChangEtAl:1997, LeaoEtAl:2004}.

\textit{Fiber bundles} are smooth manifolds endowed with a canonical map and the quotient 
topology necessary to admit canonical disintegrations and, consequently, the geometric 
equivalent of conditional and marginal probability measures.

\subsubsection{Fiber Bundles}

A smooth fiber bundle, $\pi: Z \rightarrow Q$, combines an $\left(n + k \right)$-dimensional 
total space, $Z$, an $n$-dimensional base space, $Q$, and a smooth projection, 
$\pi$, that submerses the total space into the base space.  We will refer to a positively-oriented
fiber bundle as a fiber bundle in which both the total space and the base space are 
positively-oriented and the projection operator is orientation-preserving.

Each fiber,
\begin{equation*}
Z_{q} = \pi^{-1} \! \left( q \right),
\end{equation*}
is itself a $k$-dimensional manifold isomorphic to a common fiber space, $F$, and is naturally 
immersed into the total space,
\begin{equation*}
\iota_{q} : Z_{q} \hookrightarrow Z,
\end{equation*}
where $\iota_{q}$ is the inclusion map.
We will make heavy use of the fact that there exists a trivializing cover of the base space, 
$\left\{ \mathcal{U}_{\alpha} \right\}$, along with subordinate charts and a partition of unity, 
where the corresponding total space is isomorphic to a trivial product 
(Figures \ref{fig:local_fiber_bundle}, \ref{fig:cylinder_as_bundle}),
\begin{equation*}
\pi^{-1} \! \left( \mathcal{U}_{\alpha} \right) \approx \mathcal{U}_{\alpha} \times F.
\end{equation*}

\begin{figure*}
\centering
%
\begin{tikzpicture}[scale=0.25, thick]
  
  \draw[] (1.5, 0) .. controls +(-2, 0) and +(-2, 0) .. +(5, 20);
  \draw[] (11.55, 0) .. controls +(-2, 0) and +(-2, 0) .. +(5, 20);
  
  \draw[name path=piU] (5.5, 0.6) .. controls +(-2, 0) and +(-2, 0) .. +(5, 20)
  node[above] {$\mathcal{U}_{\alpha} \times F \approx \pi^{-1} ( \mathcal{U}_{\alpha} ) \subset Z$};
  
  \draw[name path=F0, style=dashed] (1.375, -0.15) .. controls +(1, 1) and +(1, 1) .. +(10, 0);
  \draw[name path=F1] (1.1, 2.50) .. controls +(1, 1) and +(1, 1) .. +(10, 0);
  \draw[name path=F2] (1.45, 5.00) .. controls +(1, 1) and +(1, 1) .. +(10, 0);
  \draw[name path=F3] (1.89, 7.50) .. controls +(1, 1) and +(1, 1) .. +(10, 0);
  \draw[name path=F4, color=highlight] (2.500, 10.0) 
  node[left, color=black] {$Z_{q} = \pi^{-1} \! \left( q \right) \approx F$} 
  .. controls +(1, 1) and +(1, 1) .. +(10, 0);
  \draw[name path=F5] (3.15, 12.5) .. controls +(1, 1) and +(1, 1) .. +(10, 0);
  \draw[name path=F6] (3.95, 15.0) .. controls +(1, 1) and +(1, 1) .. +(10, 0);
  \draw[name path=F7] (4.85, 17.5) .. controls +(1, 1) and +(1, 1) .. +(10, 0);
  \draw[name path=F8, style=dashed] (6.35, 20) .. controls +(1, 1) and +(1, 1) .. +(10, 0);
  
  \fill[name intersections={of=piU and F4}, highlight] (intersection-1) circle (8pt);
  
  \draw[->] (20, 10) -- node[above] {$\pi$} +(10, 0);
  
  \draw[name path=U] (36.5, 0.6) .. controls +(-2, 0) and +(-2, 0) .. +(5, 20)
  node[above] {$\mathcal{U}_{\alpha} \subset Q$};
  \draw[style=dashed] (40.25, 20.45) .. controls +(1, 0.25) .. +(2, 0);
  \draw[style=dashed] (35.5, 0.45) .. controls +(1, 0.25) .. +(2, 0);

  \fill[color=highlight] (37.35, 10) circle (8pt)
  node[left, color=black] {$q$};

\end{tikzpicture}
\caption{In a local neighborhood, the total space of a fiber bundle, 
$\pi^{-1} ( \mathcal{U}_{\alpha} ) \subset Z$, is equivalent to attaching a copy of some common 
fiber space, $F$, to each point of the base space, $q \in \mathcal{U}_{\alpha} \subset Q$.  Under 
the projection operator each fiber projects back to the point at which it is attached.}
\label{fig:local_fiber_bundle}
\end{figure*}
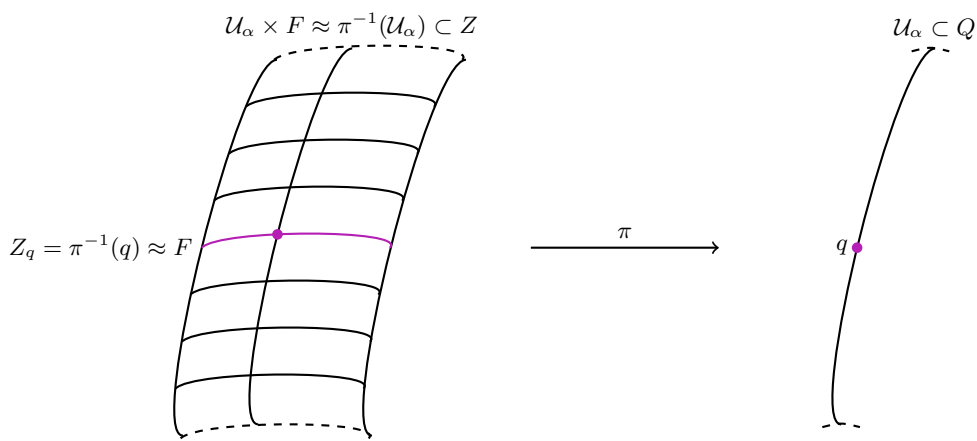

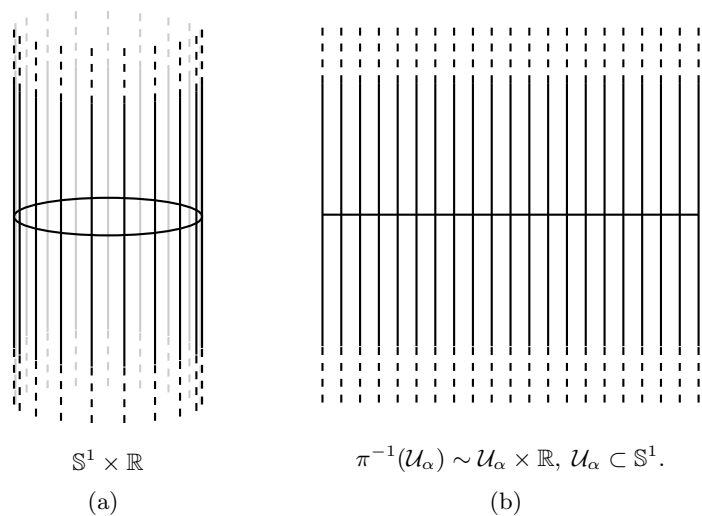
\begin{figure*}
\centering
\subfigure[]{
\begin{tikzpicture}[scale=0.25, thick]
  
 \draw[color=white] (-10, 0) -- (10, 0);
  
 \foreach \x in {0,20,...,180} {
   \draw[color=gray80, style=dashed] ({5 * cos(\x + 10)}, {-10 + sin(\x)}) -- ({5 * cos(\x + 10)}, {-7 + sin(\x)});
   \draw[color=gray80] ({5 * cos(\x + 10)}, {-7 + sin(\x)}) -- ({5 * cos(\x + 10)}, {7 + sin(\x)});
   \draw[color=gray80, style=dashed] ({5 * cos(\x + 10)}, {7 + sin(\x)}) -- ({5 * cos(\x + 10)}, {10 + sin(\x)});
 }
  
 \draw[] (0, 0) ellipse (5 and 1);

 \foreach \x in {0,20,...,180} {
   \draw[style=dashed] ({5 * cos(\x)}, {-10 - sin(\x)}) -- ({5 * cos(\x)}, {-7 - sin(\x)});
   \draw[] ({5 * cos(\x)}, {-7 - sin(\x)}) -- ({5 * cos(\x)}, {7 - sin(\x)}); 
   \draw[style=dashed] ({5 * cos(\x)}, {7 - sin(\x)}) -- ({5 * cos(\x)}, {10 - sin(\x)});
 }
  
 \node at (0, -13) { $\mathbb{S}^{1} \times \mathbb{R}$ };
 
\end{tikzpicture}
}
\subfigure[]{
\begin{tikzpicture}[scale=0.25, thick]
  
  \foreach \x in {-10,...,10} {
    \draw[style=dashed] (\x, -10) -- (\x, -7);
    \draw[] (\x, -7) -- (\x, 7);
    \draw[style=dashed] (\x, 7) -- (\x, 10);
  }
  
  \draw[] (-10, 0) -- (10, -0);
  \node at (0, -13) { 
    $\pi^{-1} \! \left( \mathcal{U}_{\alpha} \right) \sim \mathcal{U}_{\alpha} \times \mathbb{R},
    \; \mathcal{U}_{\alpha} \subset \mathbb{S}^{1}.$ };
  
\end{tikzpicture}
}
\caption{(a) The canonical projection, $\pi : \mathbb{S}^{1} \times \mathbb{R} \rightarrow \mathbb{S}^{1}$, 
gives the cylinder the structure of a fiber bundle with fiber space $F = \mathbb{R}$.  (b) The domain of
each chart becomes isomorphic to the product of a neighborhood of the base space,
$\mathcal{U}_{\alpha} \subset \mathbb{S}^{1}$, and the fiber, $\mathbb{R}$.}
\label{fig:cylinder_as_bundle}
\end{figure*}

Vector fields on $Z$ are classified by their action under the projection operator.  Vertical
vector fields, $Y_{i}$, lie in the kernel of the projection operator,
\begin{equation*}
\pi_{*} Y_{i} = 0,
\end{equation*}
while horizontal vector fields, $\tilde{X}_{i}$, pushforward to the tangent space of the
base space,
\begin{equation*}
\pi_{*} \tilde{X}_{i} \! \left( z \right) = X_{i} \! \left( \pi \! \left( z \right) \right) \in T_{\pi \left( z \right)} Q,
\end{equation*}
where $z \in Z$ and $\pi \! \left( z \right) \in Q$.  Horizontal forms are forms on the total
space that vanish then contracted against one or more vertical vector fields.

Note that vector fields on the base space do not uniquely define horizontal vector fields on 
the total space; a choice of $\tilde{X}_{i}$  consistent with $X_{i}$ is called a \textit{horizontal lift} 
of $X_{i}$.  More generally we will refer to the \textit{lift} of an object on the base space as the 
selection of some object on the total space that pushes forward to the corresponding object on the 
base space.


\subsubsection{Disintegrating Fiber Bundles}

Because both $Z$ and $Q$ are both smooth manifolds, and hence Radon spaces, the structure 
of the fiber bundle guarantees the existence of disintegrations with respect to the projection 
operator~\citep{LeaoEtAl:2004, Simmons:2012, CensorEtAl:2014} and, under certain regularity 
conditions, regular conditional probability measures.  A substantial benefit of working with smooth 
manifolds is that we can not only prove the existence of disintegrations but also explicitly construct 
their geometric equivalents and utilize them in practice.

\begin{definition}
\label{def:disintegration}

Let $\left( R, \mathcal{B} \! \left( R \right) \right)$ and $\left( S, \mathcal{B} \! \left( S \right) \right)$ 
be two measurable spaces with the respective $\sigma$-finite measures $\mu_{R}$ and $\mu_{S}$, 
and a measurable map, $F : R \rightarrow S$, between them. A \emph{disintegration} of  $\mu_{R}$ 
with respect to $F$ and $\mu_{S}$ is a map,
\begin{equation*}
\nu : S \times \mathcal{B} \! \left( R \right) \rightarrow \mathbb{R}^{+},
\end{equation*}
such that
\begin{enumerate}[i]
  \item $\nu \! \left( s, \cdot \right)$ is a $\mathcal{B} \! \left( R \right)$-finite measure
           concentrating on on the level set $F^{-1} \! \left( s \right)$, i.e. for 
           $\mu_{S}$-almost all s
           \begin{equation*}
           \nu \! \left( s, A \right) = 0, \, 
           \forall A \in \mathcal{B} \! \left( R \right) | A \cap F^{-1} \! \left( s \right) = 0,
           \end{equation*}
\end{enumerate}
and for any positive, measurable function $f \in L^{1} \! \left( R, \mu_{R} \right)$,
\begin{enumerate}[i]
  \setcounter{enumi}{1}
  \item $s \mapsto \int_{R} f \! \left( r \right) \, \nu \! \left( s, \dd r \right)$ is a measurable function 
           for all $s \in S$.
  \item $\int_{R} f \! \left( r \right) \, \mu_{R} \! \left( \dd r \right) 
           = \int_{S} \int_{F^{-1} \left( s \right)} f \! \left( r \right) 
              \nu \! \left( s, \dd r \right) \mu_{S} \! \left( \dd s \right) $.
\end{enumerate}

\end{definition}

In other words, a disintegration is an unnormalized Markov kernel that concentrates
on the level sets of $F$ instead of the whole of $R$ (Figure \ref{fig:disintegration}).  
Moreover, if $\mu_{R}$ is finite or $F$ is proper then the pushforward measure,
\begin{align*}
\mu_{S} &= T_{*} \mu_{R}
\\
\mu_{S} \! \left( B \right) 
&=
\mu_{R} \! \left( F^{-1} \! \left( B \right) \right), \, \forall B \in \mathcal{B} \! \left( S \right),
\end{align*}
is $\sigma$-finite and known as the marginalization of $\mu_{R}$ with respect to $F$.
In this case the disintegration of $\mu_{R}$ with respect to its pushforward measure 
becomes a normalized kernel and exactly a regular conditional probability measure.
The classic marginalization paradoxes of measure theory~\citep{DavidEtAl:1973} occur 
when the pushforward of $\mu_{R}$ is not $\sigma$-finite and the corresponding 
disintegration, let alone a regular conditional probability measure, does not exist;
we will be careful to explicitly exclude such cases here.

\begin{figure}
\centering
\subfigure[]{ \includegraphics[width=1.75in]{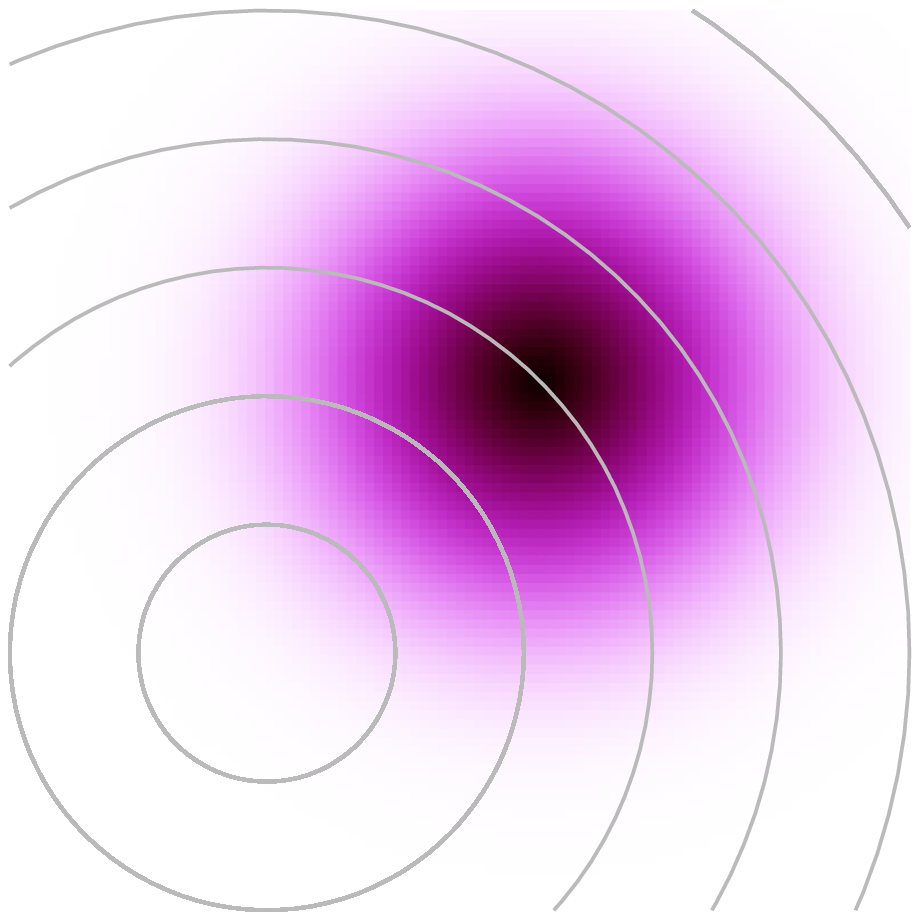} }
\subfigure[]{ \includegraphics[width=1.75in]{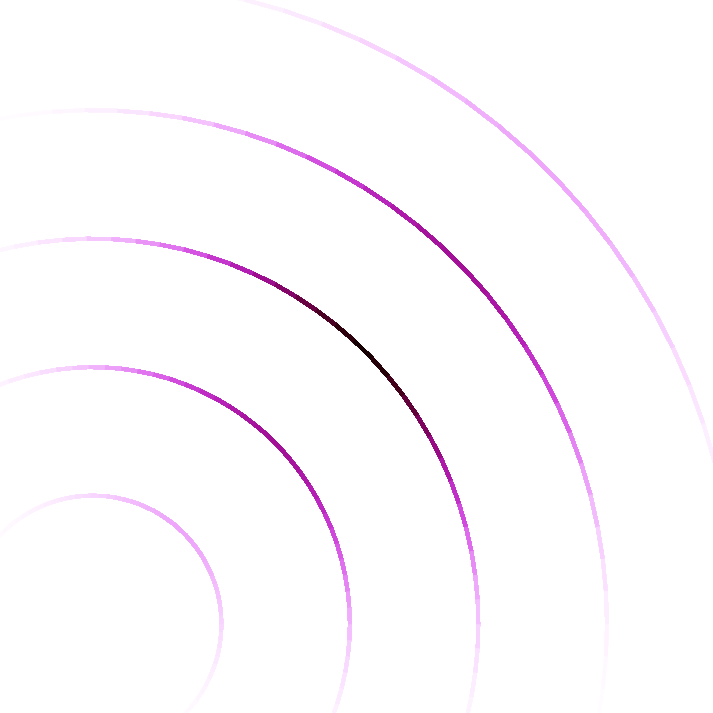} }
\caption{The disintegration of (a) a $\sigma$-finite measure on the space $R$ with respect to a 
map $F: R \rightarrow S$ and a $\sigma$-finite measure on $S$ defines (b) a family of 
$\sigma$-finite measures that concentrate on the level sets of $F$.}
\label{fig:disintegration}
\end{figure}

For the smooth manifolds of interest we do not need the full generality of disintegrations, 
and instead consider the equivalent object restricted to smooth measures.

\begin{definition}
\label{def:smooth_disintegration}

Let $R$ and $S$ be two smooth, orientable manifolds with the respective smooth measures 
$\mu_{R}$ and $\mu_{S}$, and a smooth, orientation-preserving map, $F : R \rightarrow S$, 
between them. A \emph{smooth disintegration} of  $\mu_{R}$ with respect to $F$ and 
$\mu_{S}$ is a map,
\begin{equation*}
\nu : S \times \mathcal{B} \! \left( R \right) \rightarrow \mathbb{R}^{+},
\end{equation*}
such that
\begin{enumerate}[i]
  \item $\nu \! \left( s, \cdot \right)$ is a smooth measure
           concentrating on on the level set $F^{-1} \! \left( s \right)$, i.e. for 
           $\mu_{S}$-almost all s
           \begin{equation*}
           \nu \! \left( s, A \right) = 0, \, 
           \forall A \in \mathcal{B} \! \left( R \right) | A \cap F^{-1} \! \left( s \right) = 0,
           \end{equation*}
\end{enumerate}
and for any positive, smooth function $f \in L^{1} \! \left( R, \mu_{R} \right)$,
\begin{enumerate}[i]
  \setcounter{enumi}{1}
  \item The function $F \! \left( s \right) = \int_{R} f \! \left( r \right) \, \nu \! \left( s, \dd r \right)$ 
           is integrable with respect to any smooth measure on $S$. 
  \item $\int_{R} f \! \left( r \right) \, \mu_{R} \! \left( \dd r \right) 
           = \int_{S} \int_{F^{-1} \left( s \right)} f \! \left( r \right) 
              \nu \! \left( s, \dd r \right) \mu_{S} \! \left( \dd s \right) $.
\end{enumerate}

\end{definition}

Smooth disintegrations have a particularly nice geometric interpretation: Definition \ref{def:smooth_disintegration}i
implies that disintegrations define volume forms when pulled back onto the fibers, while Definition
\ref{def:smooth_disintegration}ii implies that the volume forms are smoothly immersed into the total space
(Figure \ref{fig:geometric_disintegration}).  Hence if we want to construct smooth disintegrations geometrically
then we should consider the space of $k$-forms on $Z$ that restrict to finite volume forms on the 
fibers, i.e. $\omega \in \Omega^{k} \! \left( Z \right)$ satisfying
\begin{align*}
\iota_{q}^{*} \omega &> 0
\\
\int_{Z_{q}} \iota_{q}^{*} \omega &< \infty.
\end{align*}
Note that the finiteness condition is not strictly necessary, but allows us to construct
smooth disintegrations independent of the exact measure being disintegrated.

\begin{figure*}
\centering
\subfigure[]{
\begin{tikzpicture}[scale=0.20, thick]
  
  \foreach \x in {2,4,...,98} {
    \draw[style=dashed, color=gray80] ({\x / 5 - 10}, -10) -- ({\x / 5 - 10}, -7);
    \draw[color=gray80] ({\x / 5 - 10}, -7) -- ({\x / 5 - 10}, 7);
    \draw[style=dashed, color=gray80] ({\x / 5 - 10}, 7) -- ({\x / 5 - 10}, 10);
  }
  
  \draw[color=gray80] (-10, 0) -- (10, -0);
  
  \node[] at (0,0) {\includegraphics[width=3.96cm]{cylinder_density_flat.eps}};

  \node at (0, -12.5) { 
    $\pi^{-1} \! \left( \mathcal{U}_{\alpha} \right) \sim \mathcal{U}_{\alpha} \times \mathbb{R}$
  };
  
\end{tikzpicture}
}
\subfigure[]{
\begin{tikzpicture}[scale=0.20, thick]
  
  \fill[color=white] (-10, -10) rectangle (10, 12);
  
  \foreach \x in {2,4,...,98} {
    \draw[color=gray80] ({\x / 5 - 10}, -6) -- (({\x / 5 - 10}, -5.5);
  }
  
  \draw[] (-10, -6) -- (10, -6);
  
  \node[] at (0, 0) {\includegraphics[width=4cm]{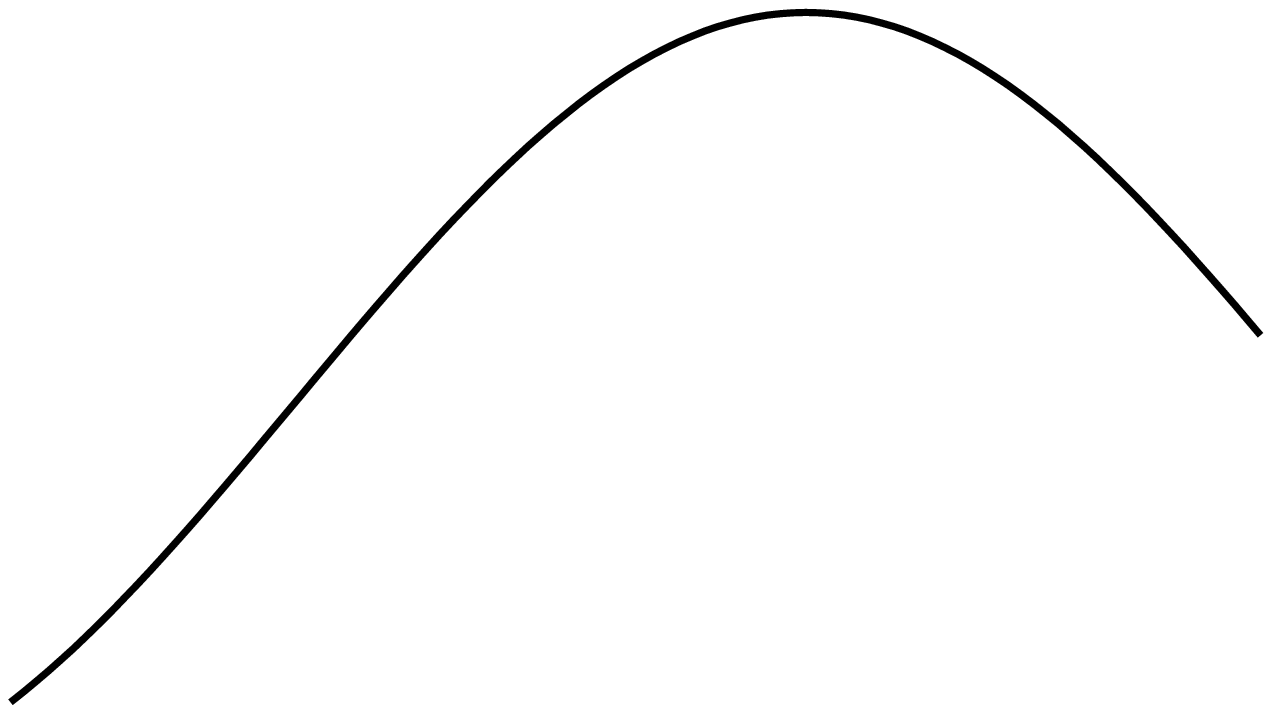}};
  
  \node at (0, -10.5) { $\mathcal{U}_{\alpha} \subset \mathbb{S}^{1}$};
  
\end{tikzpicture}
}
\subfigure[]{
\begin{tikzpicture}[scale=0.20, thick]
  
  \foreach \x in {2,4,...,98} {
    \draw[style=dashed, color=gray80] ({\x / 5 - 10}, -10) -- ({\x / 5 - 10}, -7);
    \draw[color=gray80] ({\x / 5 - 10}, -7) -- ({\x / 5 - 10}, 7);
    \draw[style=dashed, color=gray80] ({\x / 5 - 10}, 7) -- ({\x / 5 - 10}, 10);
  }
  
  \draw[color=gray80] (-10, 0) -- (10, -0);
  
  \node[] at (-0.1, 0) {\includegraphics[width=3.96cm]{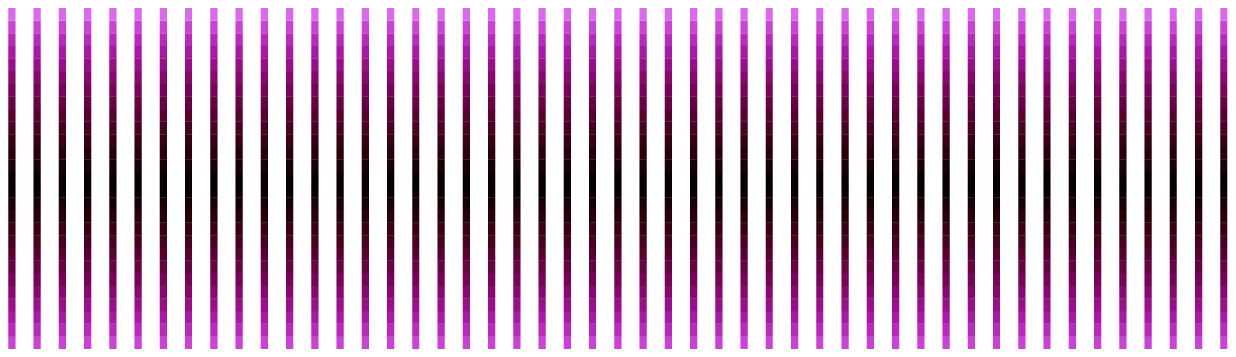}};
  
  \node at (0, -12.5) { 
    $\pi^{-1} \! \left( \mathcal{U}_{\alpha} \right) \sim \mathcal{U}_{\alpha} \times \mathbb{R}$
  };
  
\end{tikzpicture}
}
\caption{Considering the cylinder as a fiber bundle, 
$\pi : \mathbb{S}^{1} \times \mathbb{R} \rightarrow \mathbb{S}^{1}$,
(a) any joint measure on the total space and (b) any measure on the base 
space define (c) a disintegration that concentrates on the fibers.  Given
any two of these objects we can uniquely construct the third.}
\label{fig:geometric_disintegration}
\end{figure*}

The only subtlety with such a definition is that $k$-forms on the total space differing by only a
horizontal form will restrict to the same volume form on the fibers.  Consequently we will
consider the equivalence classes of $k$-forms up to the addition of horizontal fields,
\begin{align*}
\Upsilon \! \left( \pi : Z \rightarrow Q \right) &\subset \Omega^{k} \! \left( Z \right) / \sim
\\
\omega_{1} \sim \omega_{2} &\Leftrightarrow \omega_{1} - \omega_{2} 
\in \Omega^{k}_{H} \! \left( \pi : Z \rightarrow Q \right),
\end{align*}
where $\Omega^{k}_{H} \! \left( \pi : Z \rightarrow Q \right)$ is the space of horizontal $k$-forms
on the total space, with elements $\upsilon \in \Upsilon \! \left( \pi : Z \rightarrow Q \right)$
satisfying
\begin{align*}
\iota_{q}^{*} \upsilon &> 0
\\
\int_{Z_{q}} \iota_{q}^{*} \upsilon &< \infty.
\end{align*}
As expected from the fact that any smooth manifold is a Radon space, such forms always
exist.

\begin{lemma}
\label{lem:existence_of_fiber_volume_forms}

The space $\Upsilon \! \left( \pi : Z \rightarrow Q \right)$ is convex and nonempty.

\end{lemma}

Given a point on the base space, $q \in Q$, the elements of $\Upsilon \! \left( \pi : Z \rightarrow Q \right)$
naturally define smooth measures that concentrate on the fibers.

\begin{lemma}
\label{lem:fiber_forms_as_disintegrating_kernels}

Any element of $\Upsilon \! \left( \pi : Z \rightarrow Q \right)$ defines a smooth measure,
\begin{align*}
\nu :
& \; Q \times \mathcal{B} \! \left( Z \right) \rightarrow \mathbb{R}^{+}
\\
& \; q, A \mapsto \int_{\iota_{q} \left( A \, \cap \, Z_{q} \right) } \iota_{q}^{*} \upsilon,
\end{align*}
concentrating on the fiber $Z_{q}$, 
\begin{equation*}
\nu \! \left( q, A \right) = 0, \,
\forall A \in \mathcal{B} \! \left( Z \right) | A \cap Z_{q} = 0.
\end{equation*}

\end{lemma}

Finally, any $\upsilon \in \Upsilon \! \left( \pi : Z \rightarrow Q \right)$ satisfies an equivalent of the product rule.

\begin{lemma}
\label{lem:measure_lift}

Any element $\upsilon \in \Upsilon \! \left( \pi : Z \rightarrow Q \right)$ lifts any smooth measure
on the base space, $\mu_{Q} \in \mathcal{M} \! \left( Q \right)$, to a smooth measure on the total 
space by
\begin{equation*}
\mu_{Z} = \pi^{*} \mu_{Q} \wedge \upsilon \in \mathcal{M} \! \left( Z \right).
\end{equation*}

\end{lemma}

\noindent Note the resemblance to the typical measure-theoretic result,
\begin{equation*}
\mu_{Z} ( \dd z ) = \mu_{Q} ( \dd q ) \, \upsilon ( q, \dd z ).
\end{equation*}

Consequently, the elements of $\Upsilon \! \left( \pi : Z \rightarrow Q \right)$ define 
smooth disintegrations of any smooth measure on the total space.

\begin{theorem}
\label{thm:existence_of_disintegrations}

A positively oriented, smooth fiber bundle admits a smooth disintegration of any smooth 
measure on the total space, $\mu_{Z} \in \mathcal{M} \! \left( Z \right)$, with respect to the 
projection operator and any smooth measure $\mu_{Q} \in \mathcal{M} \! \left( Q \right)$.

\end{theorem}

\begin{proof}

From Lemma \ref{lem:measure_lift} we know that for any 
$\upsilon' \in \Upsilon \! \left( \pi : Z \rightarrow Q \right)$ the exterior product 
$\pi^{*} \mu_{Q} \wedge \upsilon'$ is a smooth measure on the total space, and 
because the space of smooth measures is one-dimensional we must have
\begin{equation*}
\mu_{Z} = g \, \pi^{*} \mu_{Q}  \wedge \upsilon',
\end{equation*}
for some bounded, positive function $g : Z \rightarrow \mathbb{R}^{+}$.  Because 
$g$ is everywhere positive it can be absorbed into $\upsilon$ to define a new, unique 
element $\upsilon \in \Upsilon \! \left( Z \right)$ such that 
$\mu_{Z} = \pi^{*} \mu_{Q} \wedge \upsilon$.  With Lemma 
\ref{lem:existence_of_fiber_volume_forms} showing that 
$\Upsilon \! \left( \pi : Z \rightarrow Q \right)$ is non-empty, such an $\upsilon$ exists 
for any positively-oriented, smooth fiber bundle.

From Lemma \ref{lem:fiber_forms_as_disintegrating_kernels}, this $\upsilon$ defines a 
smooth kernel and, hence, satisfies Definition \ref{def:smooth_disintegration}i.

If $\lambda_{Q}$ is any smooth measure on $Q$, not necessarily equal to $\mu_{Q}$, then
for any smooth, positive function $f \in L^{1} \! \left( Z, \mu_{Z} \right)$ we have
\begin{align*}
\int_{Q} F \! \left( q \right) \lambda_{Q}
&=
\int_{Q} \left[ \int_{Z_{q}} \iota_{q}^{*} \! \left( f \, \upsilon \right) \right] \lambda_{Q},
\end{align*}
or, employing a trivializing cover,
\begin{align*}
\int_{Q} F \! \left( q \right) \lambda_{Q}
&=
\sum_{\alpha} \int_{\mathcal{U}_{\alpha}} \rho_{\alpha} \left[ \int_{Z_{q}} \iota_{q}^{*} \! \left( f_{\alpha} \, \upsilon \right) \right] \lambda_{Q}
\\
&=
\sum_{\alpha} \int_{\mathcal{U}_{\alpha}} \rho_{\alpha} \left[ \int_{F} \iota_{q}^{*} \! \left( f_{\alpha} \, \upsilon \right) \right] \lambda_{Q}
\\
&=
\sum_{\alpha} \int_{\mathcal{U}_{\alpha} \times F} \rho_{\alpha} \, f_{\alpha} \, \pi^{*} \lambda_{Q} \wedge \upsilon.
\end{align*}
Once again noting that the space of smooth measures on $Z$ is one-dimensional, we must have for some positive, bounded
function $g : Z \rightarrow \mathbb{R}^{+}$,
\begin{align*}
\int_{Q} F \! \left( q \right) \lambda_{Q}
&=
\sum_{\alpha} \int_{\mathcal{U}_{\alpha} \times F} \rho_{\alpha} \, f_{\alpha} \, g \, \mu_{Z}
\\
&=
\int_{Z} f \, g \, \mu_{Z}.
\end{align*}
Because $g$ is bounded the integral is finite given the $\mu_{Z}$ integrability of $f$, hence $\upsilon$ satisfies Definition \ref{def:smooth_disintegration}ii.

Similarly, for any smooth, positive function $f \in L^{1} \! \left( Z, \mu_{Z} \right)$,
\begin{align*}
\int_{Z} f  \, \mu_{Z}
&= 
\int_{Z} f  \, \pi^{*} \mu_{Q} \wedge \upsilon
\\
&=
\sum_{\alpha} \int_{\mathcal{U}_{\alpha} \times F }
\rho_{\alpha} \, f_{\alpha} \, \pi^{*} \mu_{Q} \wedge \upsilon
\\
&=
\sum_{\alpha} \int_{\mathcal{U}_{\alpha}} \rho_{\alpha}
\left[ \int_{F} \iota^{*}_{q} \! \left( f_{\alpha} \upsilon \right) \right]
\mu_{Q}
\\
&=
\int_{Q} \left[ \int_{Z_{q}} \iota^{*}_{q} \! \left( f \,  \upsilon \right) \right]
\mu_{Q}
\\
&=
\int_{Q} \left[ \int_{Z_{q}} f \, \nu \! \left( q, \cdot \right) \right]
\mu_{Q}.
\end{align*}
Because of the finiteness of $\upsilon$ the integral is well-defined and the kernel satisfies Definition 
\ref{def:smooth_disintegration}iii.

Hence for any smooth fiber bundle and smooth measures $\mu_{Z}$ and $\mu_{Q}$, 
there exists an $\upsilon \in \Upsilon \! \left( \pi : Z \rightarrow Q \right)$ that induces a 
smooth disintegration of $\mu_{Z}$ with respect to the projection operator and $\mu_{Q}$.

\end{proof}

Ultimately we're not interested in smooth disintegrations but rather regular conditional probability measures.
Fortunately, the elements of $\Upsilon \! \left( \pi : Z \rightarrow Q \right)$ are only a normalization away
from defining the desired probability measures.  To see this, first note that we can immediately define the new 
space $\Xi \! \left( \pi : Z \rightarrow Q \right)$ with elements $\xi \in \Xi \! \left( \pi : Z \rightarrow Q \right)$ satisfying
\begin{align*}
\iota_{q}^{*} \xi &> 0,
\\
\int_{Z_{q}} \iota_{q}^{*} \xi &= 1.
\end{align*}
The elements of $\Xi \! \left( \pi : Z \rightarrow Q \right)$ relate a smooth measure on the total space 
to it's pushforward measure with respect to the projection, provided it exists, which is exactly the 
property needed for the smooth disintegrations to be regular conditional probability measures.

\begin{lemma}
\label{lem:pushforward_measures}

Let $\mu_{Z}$ be a smooth measure on the total space of a positively-oriented, smooth fiber bundle 
with $\mu_{Q}$ the corresponding pushforward measure with respect to the projection operator, 
$\mu_{Q} = \pi_{*} \mu_{Z}$.  If  $\mu_{Q}$ is a smooth measure then 
$\mu_{Z} = \pi^{*} \mu_{Q} \wedge \xi$ for a unique element of $\xi \in \Xi \! \left( \pi : Z \rightarrow Q \right)$.

\end{lemma}

Consequently the elements of $\Xi \! \left( \pi : Z \rightarrow Q \right)$ also define regular conditional
probability measures.

\begin{theorem}
\label{thm:existence_of_conditional_measures}

Any smooth measure on the total space of a positively-oriented, smooth fiber bundle admits a 
regular conditional probability measure with respect to the projection operator provided
that the pushforward measure with respect to the projection operator is smooth.

\end{theorem}

\begin{proof}

From Lemma \ref{lem:pushforward_measures} we know that for any smooth measure $\mu_{Z}$ 
there exists a $\xi \in \Xi \! \left( E \right)$ such that $\mu_{Z} = \pi^{*} \mu_{Q} \wedge \xi$ so long
as the pushforward measure, $\mu_{Q}$, is smooth.  Applying Theorem \ref{thm:existence_of_disintegrations}, 
any choice of $\xi$ then defines a smooth disintegration of $\mu_{Z}$ with respect to the projection operator
and the pushforward measure and hence the disintegration is a regular conditional probability measure.

\end{proof}

Although we have shown that elements of $\Xi \! \left( \pi : Z \rightarrow Q \right)$ disintegrate 
measures on fiber bundles, we have not yet explicitly constructed them.  Fortunately the fiber 
bundle geometry proves productive here as well.


\subsubsection{Constructing Smooth Measures From Smooth Disintegrations}

The geometric construction of regular conditional probability measures is particularly valuable
because it provides an explicit construction for lifting measures on the base space to measures
on the total space as well as marginalizing measures on the total space down to the base space. 

As shown above, the selection of any element of $\Xi \! \left( \pi : Z \rightarrow Q \right)$ defines 
a lift of smooth measures on the base space to smooth measures on the total space.

\begin{corollary}
\label{cor:joint_from_marginal}

If $\mu_{Q}$ is a smooth measure on the base space of a positively-oriented, smooth fiber 
bundle then for any $\xi \in \Xi \! \left( \pi : Z \rightarrow Q \right)$, 
$\mu_{Z} = \pi^{*} \mu_{Q} \wedge \xi$ is a smooth measure on the total space whose 
pushforward is $\mu_{Q}$.

\end{corollary}

\begin{proof}

$\mu_{Z} = \pi^{*} \mu_{Q} \wedge \xi$ is a smooth measure on the total space by Lemma 
\ref{lem:measure_lift}, and Lemma \ref{lem:pushforward_measures} immediately implies that its 
pushforward is $\mu_{Q}$.

\end{proof}

Even before constructing the pushforward measure of a measure on the total space, we can construct 
its regular conditional probability measure with respect to the projection.

\begin{lemma}
\label{lem:canonical_disintegration}

Let $\mu_{Z}$ be a smooth measure on the total space of a positively-oriented, smooth fiber bundle
whose pushforward measure with respect to the projection operator is smooth, with $U \subset Q$
any neighborhood of the base space that supports a local frame.  Within $\pi^{-1} \! \left( U \right)$, 
the element $\xi \in \Xi \! \left( \pi : Z \rightarrow Q \right)$
\begin{equation*}
\xi = 
\frac{ \left( \tilde{X}_{1}, \ldots, \tilde{X}_{n} \right) \, \lrcorner \, \mu_{Z} }
{ \mu_{Q} \! \left( X_{1}, \ldots, X_{n} \right) },
\end{equation*}
defines the regular conditional probability measure of $\mu_{Z}$ with respect to the projection operator,
where $\left( X_{1}, \ldots, X_{n} \right)$ is any positively-oriented frame in $U$ satisfying
\begin{equation*}
\mu_{Q} \! \left( X_{1}, \ldots, X_{n} \right) < \infty, \, \forall q \in U
\end{equation*}
and $\left( \tilde{X}_{1}, \ldots, \tilde{X}_{n} \right)$ is any corresponding horizontal lift.

\end{lemma}

The regular conditional probability measure then allows us to validate the geometric construction of the 
pushforward measure.

\begin{corollary}
\label{cor:construction_of_pushforward_measure}

Let $\mu_{Z}$ be a smooth measure on the total space of a positively-oriented, smooth fiber bundle
whose pushforward measure with respect to the projection operator is smooth, with $U \subset Q$
any neighborhood of the base space that supports a local frame.  The pushforward measure
at any $q \in U$ is given by
\begin{equation*}
\mu_{Q} \! \left( X_{1} \! \left( q \right), \ldots, X_{n} \! \left( q \right) \right) = 
\int_{Z_{q}} \iota^{*}_{q} \left( \left( \tilde{X}_{1}, \ldots, \tilde{X}_{n} \right)
\, \lrcorner \, \mu_{Z} \right),
\end{equation*}
where $\left( X_{1}, \ldots, X_{n} \right)$ is any positively-oriented frame in $U$ satisfying
\begin{equation*}
\mu_{Q} \! \left( X_{1}, \ldots, X_{n} \right) < \infty, \, \forall q \in U
\end{equation*}
and $\left( \tilde{X}_{1}, \ldots, \tilde{X}_{n} \right)$ is any corresponding horizontal lift.

\end{corollary}

\begin{proof}

From Lemma \ref{lem:canonical_disintegration}, the regular conditional probability measure of 
$\mu_{Z}$ with respect to the projection operator is defined by
\begin{equation*}
\xi = 
\frac{ \left( \tilde{X}_{1}, \ldots, \tilde{X}_{n} \right) \, \lrcorner \, \mu_{Z} }
{ \mu_{Q} \! \left( X_{1}, \ldots, X_{n} \right) }.
\end{equation*}
By construction $\xi$ restricts to a unit volume form on any fiber within $U$, hence
\begin{align*}
1 
&=
\int_{Z_{q}} \iota_{q}^{*} \xi
\\
&= 
\frac{ \int_{Z_{q}} \iota_{q}^{*} \left( \left( \tilde{X}_{1}, \ldots, \tilde{X}_{n} \right) \, \lrcorner \, \mu_{Z} \right) }
{ \mu_{Q} \! \left( X_{1} \! \left( q \right), \ldots, X_{n} \! \left( q \right) \right) },
\end{align*}
or
\begin{equation*}
\mu_{Q} \! \left( X_{1} \! \left( q \right), \ldots, X_{n} \! \left( q \right) \right) = 
\int_{Z_{q}} \iota^{*}_{q} \left( \left( \tilde{X}_{1}, \ldots, \tilde{X}_{n} \right)
\, \lrcorner \, \mu_{Z} \right),
\end{equation*}
as desired.

\end{proof}


\subsection{Measures on Riemannian Manifolds}

Once the manifold is endowed with a Riemannian metric, $g$, the constructions 
considered above become equivalent to results in classical 
\textit{geometric measure theory}~\citep{Federer1969}.

In particular, the rigid structure of the metric defines projections, and hence regular 
conditional probability measures, onto any submanifold.  The resulting conditional
and marginal measures are exactly the co-area and area measures of geometric
measure theory.

Moreover, the metric defines a canonical volume form, $V_{g}$, on the manifold,
\begin{equation*}
V_{g} = 
\sqrt{ \left| g \right| } \, \dd q^{1} \wedge \ldots \wedge \dd q^{n}.
\end{equation*}
Probabilistically, $V_{g}$ is a \textit{Hausdorff} measure that generalizes the 
Lesbegue measure on $\mathbb{R}^{n}$.  If the metric is Euclidean then the 
manifold is globally isomorphic to $\mathbb{R}^{n}$ and the Hausdorff
measure reduces to the usual Lebesgue measure.


\subsection{Measures on Symplectic Manifolds}

A symplectic manifold is an even-dimensional manifold, $M$, endowed with a
non-degenerate symplectic form, $\omega \in \Omega^{2} \! \left( M \right)$.
Unlike Riemannian metrics, there are no local invariants that distinguish between
different choices of the symplectic form: within the neighborhood of any chart all 
symplectic forms are isomorphic to each other and to canonical symplectic form,
\begin{equation*}
\omega = \sum_{i = 1}^{n} \dd q^{i} \wedge \dd p_{i},
\end{equation*}
where $\left( q^{1}, \ldots, q^{n}, p_{1}, \ldots, p_{n} \right)$ are denoted canonical
or Darboux coordinates.

From our perspective, the critical property of symplectic manifolds is that the symplectic 
form admits not only a canonical family of smooth measures but also a flow that 
preserves those measures.  This structure will be the fundamental basis of Hamiltonian
Monte Carlo and hence pivotal to a theoretical understanding of the algorithm.


\subsubsection{The Symplectic Measure} 

Wedging the non-degenerate symplectic form together,
\begin{equation*}
\Omega = \bigwedge_{i = 1}^{n} \omega,
\end{equation*}
yields a canonical volume form on the manifold.

The equivalence of symplectic forms also ensures that the symplectic volumes,
given in local coordinates as
\begin{equation*}
\Omega = n! \left(
\dd q^{1} \wedge \ldots \wedge \dd q^{n} \wedge
\dd p_{1} \wedge \ldots \wedge \dd p_{n} \right),
\end{equation*}
are also equivalent locally.


\subsubsection{Hamiltonian Systems and Canonical Measures} \label{sec:hamiltonian_systems}

A symplectic manifold becomes a Hamiltonian system with the selection of a 
smooth Hamiltonian function,
\begin{equation*}
H : M \rightarrow \mathbb{R}.
\end{equation*} 

Together with the symplectic form, a Hamiltonian defines a corresponding vector
field,
\begin{equation*}
\dd H = \omega \! \left( X_{H}, \cdot \right)
\end{equation*}
naturally suited to the Hamiltonian system.  In particular, the vector field preserves
both the symplectic measure and the Hamiltonian,
\begin{equation*}
\mathcal{L}_{X_{H}} \Omega = \mathcal{L}_{X_{H}} H = 0.
\end{equation*}
Consequently any measure of the form
\begin{equation*}
e^{-\beta H} \Omega, \, \beta \in \mathbb{R}^{+},
\end{equation*}
known collectively as \textit{Gibbs measures} or \textit{canonical distributions}~\citep{Souriau:1997}, 
is invariant to the flow generated by the Hamiltonian vector field (Figure \ref{fig:hamiltonian_flow}),
\begin{equation*}
\left( \phi^{H}_{t} \right)_{*} \! \left( e^{-\beta H} \Omega \right) = e^{-\beta H} \Omega,
\end{equation*}
where
\begin{equation*}
X_{H} = \left. \frac{ \dd \phi^{H}_{t} }{ \dd t } \right|_{t = 0}.
\end{equation*}

\begin{figure}
\centering
\includegraphics[width=3in]{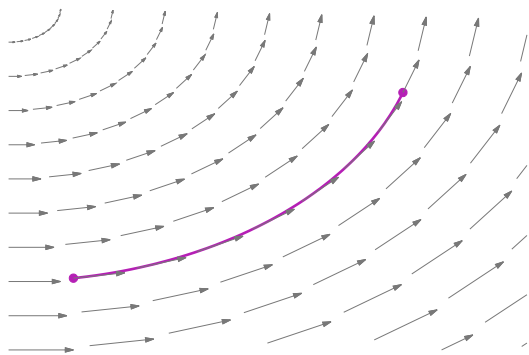}
\caption{Hamiltonian flow is given by dragging points along the integral curves of the corresponding
Hamiltonian vector field.  Because they preserve the symplectic measure, Hamiltonian vector fields are 
said to be \textit{divergenceless} and, consequently, the resulting flow preserves any canonical distribution.}
\label{fig:hamiltonian_flow}
\end{figure}

The level sets of the Hamiltonian,
\begin{equation*}
H^{-1} \! \left( E \right) = \left\{ z \in M | H \! \left( z \right) = E \right\},
\end{equation*}
decompose into \textit{regular level sets} containing only regular points of the Hamiltonian and
\textit{critical level sets} which contain at least one critical point of the Hamiltonian.  When the
critical level sets are removed from the manifold it decomposes into disconnected components,
$M = \coprod_{i} M_{i}$, each of which foliates into level sets that are diffeomorphic to some
common manifold (Figure \ref{fig:cylinder_level_sets}).  Consequently each 
$H : M_{i} \rightarrow \mathbb{R}$ becomes a smooth fiber bundle with the level sets taking the 
role of the fibers.

\begin{figure*}
\centering
%
\begin{tikzpicture}[scale=0.25, thick]
  
  \draw[] (-5, -8) -- (-5, 8);
  \draw[] (5, -8) -- (5, 8);
  \draw[style=dashed] (-5, -8) arc (180:360:5 and 1);
  \draw[style=dashed] (0, 8) ellipse (5 and 1);
  
  \node[] at (0,0) {\includegraphics[width=2.5cm]{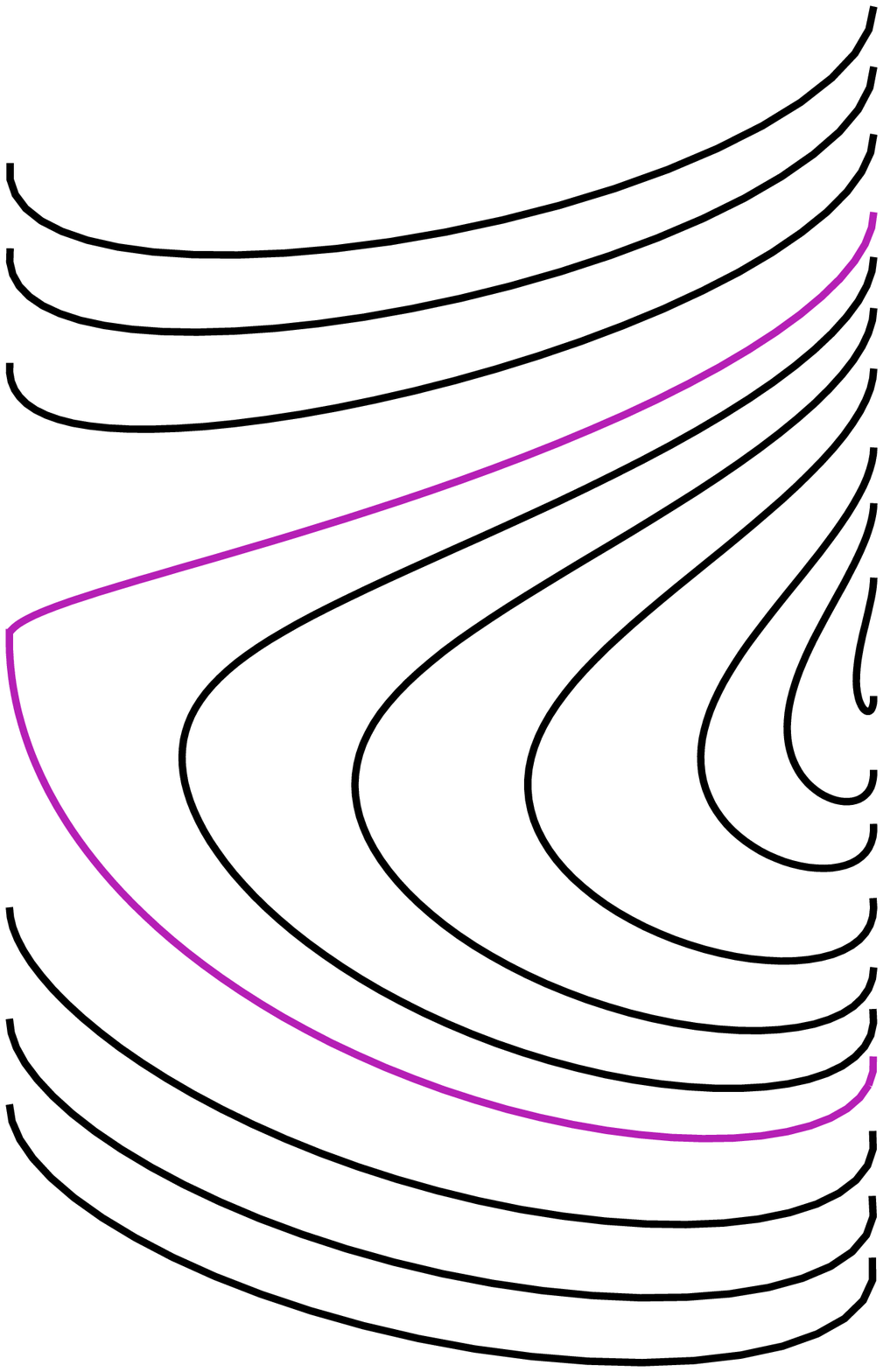}};
  \fill[color=highlight] (5, -0.33) circle (3pt);
  
  \draw[->] (6, 0) -- +(4, 0);
  
\end{tikzpicture}
%
\begin{tikzpicture}[scale=0.25, thick]
  
  \draw[] (-5, 0) -- (-5, 8);
  \draw[] (5, 4.5) -- (5, 8);
  \draw[style=dashed] (0, 8) ellipse (5 and 1);
  
  \node[] at (0,0) {\includegraphics[width=2.5cm]{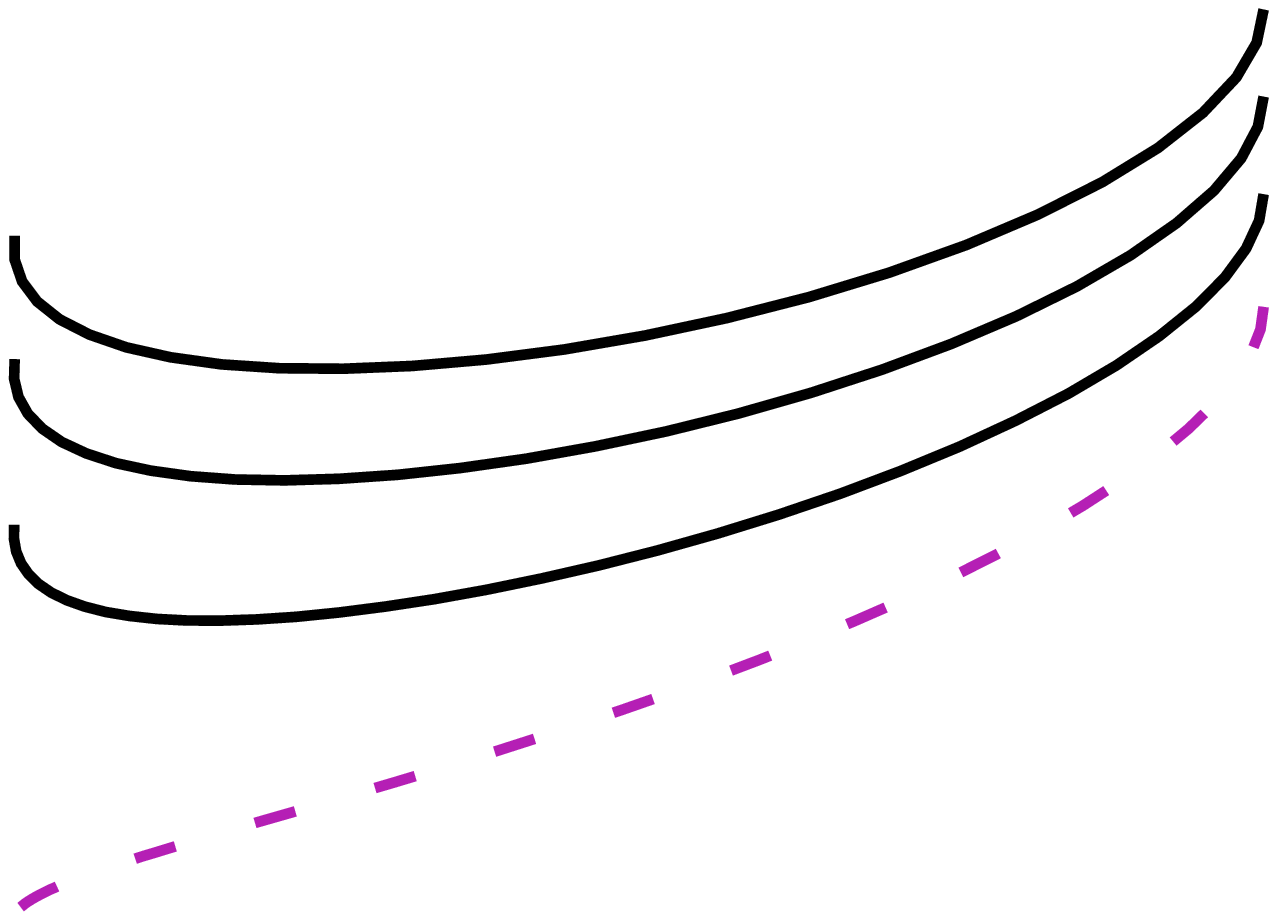}};
  
  \node[] at (7, 0) { $+$};
  
\end{tikzpicture}
%
\begin{tikzpicture}[scale=0.25, thick]
  
  \draw[] (5, -5) -- (5, 5);
  \node[] at (0,0) {\includegraphics[width=2.5cm]{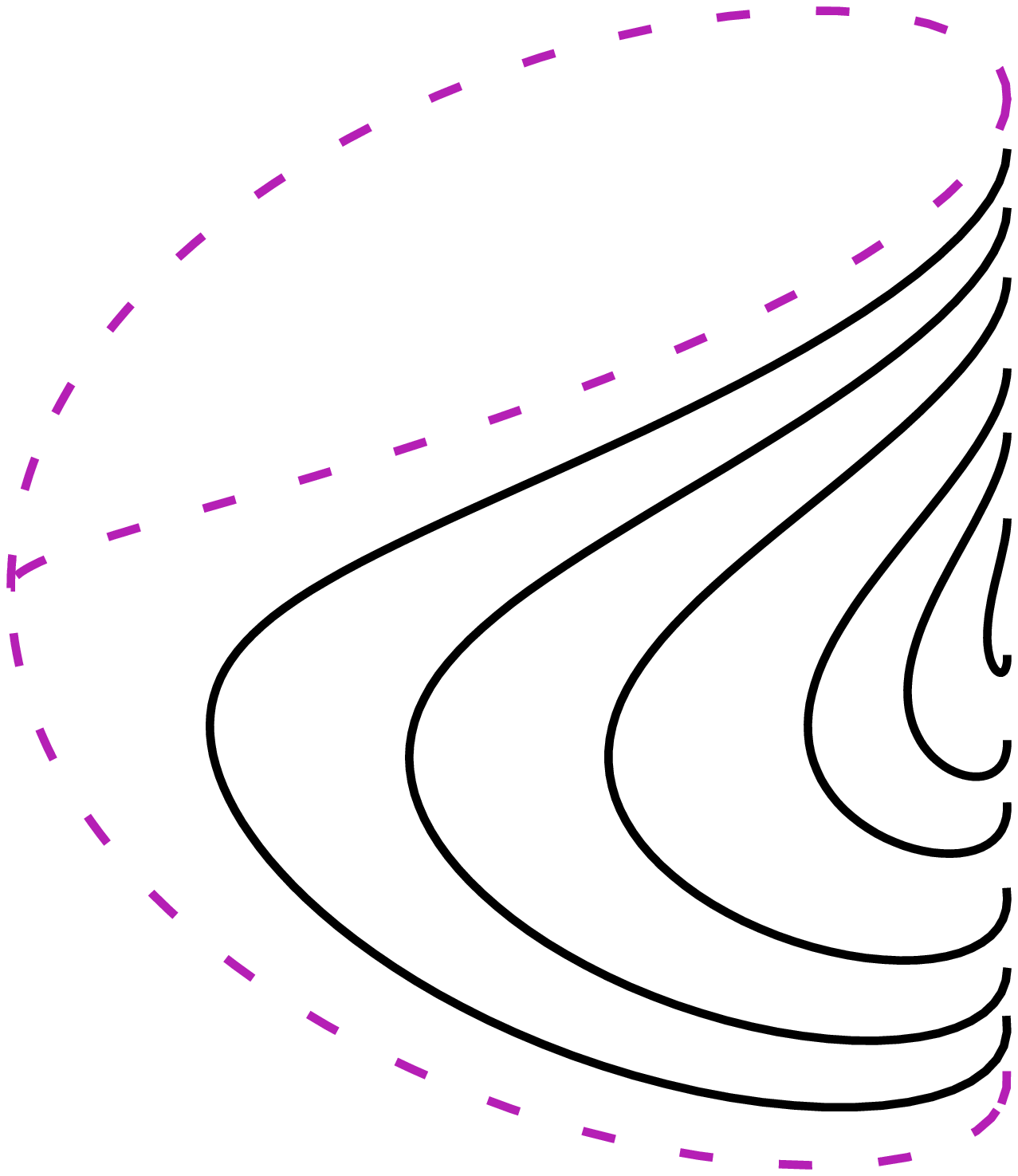}};
  
  \node[] at (7, 0) { $+$};
  
\end{tikzpicture}
%
\begin{tikzpicture}[scale=0.25, thick]
  
  \draw[] (-5, -8) -- (-5, -0.2);
  \draw[] (5, -8) -- (5, -5);
  \draw[style=dashed] (-5, -8) arc (180:360:5 and 1);
  
  \node[] at (0,0) {\includegraphics[width=2.5cm]{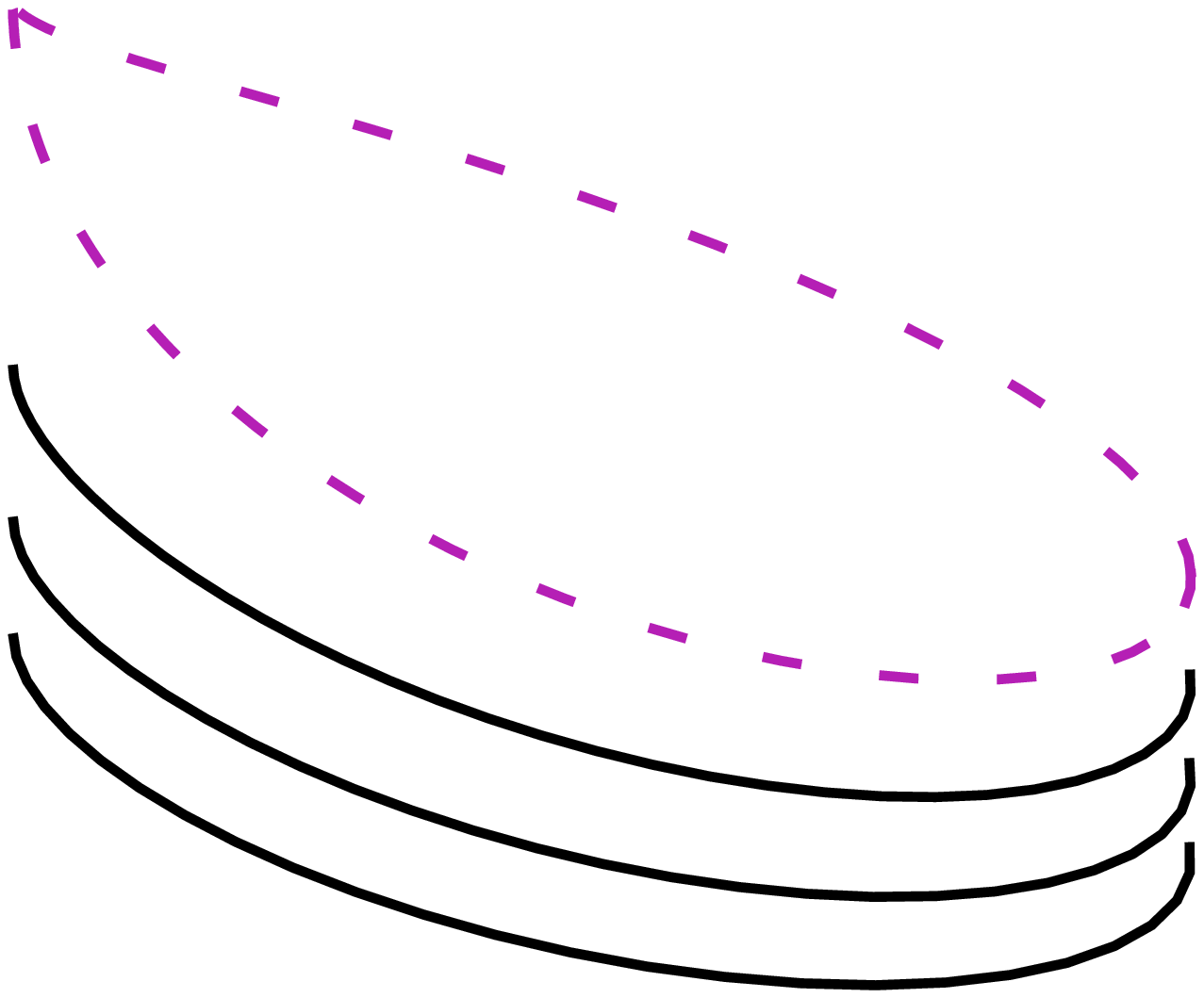}};
  
\end{tikzpicture}
\caption{When the cylinder, $\mathbb{S}^{1} \times \mathbb{R}$ is endowed with a symplectic structure
and Hamiltonian function, the level sets of a Hamiltonian function foliate the manifold.
Upon removing the critical level sets, here shown in purple, the cylinder decomposes 
into three components, each of which becomes a smooth fiber bundle with fiber space 
$F = \mathbb{S}^{1}$.}
\label{fig:cylinder_level_sets}
\end{figure*}

Provide that it is finite, 
\begin{equation*}
\int_{M} e^{- \beta H} \Omega < \infty,
\end{equation*}
upon normalization the canonical distribution becomes a probability measure,
\begin{equation*}
\varpi = \frac{ e^{-\beta H} \Omega }{ \int_{M} e^{-\beta H} \Omega },
\end{equation*}
Applying Lemma \ref{lem:canonical_disintegration}, each component of the excised canonical 
distribution then disintegrates into \textit{microcanonical distributions} on the level sets,
\begin{equation*}
\varpi_{H^{-1} (E) } = 
\frac{ v \, \lrcorner \, \Omega }
{ \int_{H^{-1} (E) } \iota^{*}_{E} \left( v \, \lrcorner \, \Omega \right) }.
\end{equation*}
Similarly, the pushforward measure on $\mathbb{R}$ is given by Lemma
 \ref{cor:construction_of_pushforward_measure},
\begin{equation*}
H_{*} \varpi = 
\frac{ e^{- \beta E} }{ \int_{M} e^{-\beta H} \Omega }
\frac{ \left( \int_{H^{-1} (E) } \iota_{E}^{*} \left( v \, \lrcorner \, \Omega \right) \right)  }
{ \dd H \! \left( v \right) } \dd E,
\end{equation*}
where $v$ is any positively-oriented horizontal vector field satisfying $\dd H \! \left( v \right) = c$ 
for some $0 < c < \infty$.  Because the critical level sets have zero measure with respect to the 
canonical distribution, the disintegration on the excised manifold defines a valid disintegration of 
original manifold as well.  For more on non-geometric constructions of the microcanonical distribution 
see~\cite{Draganescu:2009}.

The disintegration of the canonical distribution is also compatible with the Hamiltonian flow.

\begin{lemma}
\label{lem:invariance_of_microcanonical}

Let $\left( M, \Omega, H \right)$ be a Hamiltonian system with the finite and smooth canonical measure,
$\mu = e^{- \beta H} \Omega$.  The microcanonical distribution on the level set $H^{-1} (E)$,
\begin{equation*}
\varpi_{H^{-1} (E) } = 
\frac{ v \, \lrcorner \, \Omega }
{ \int_{H^{-1} (E) } \iota^{*}_{E} \left( v \, \lrcorner \, \Omega \right) },
\end{equation*}
is invariant to the corresponding Hamiltonian flow restricted to the level set, 
$\left. \phi^{H}_{t} \right|_{H^{-1} (E) }$.

\end{lemma}

The density of the pushforward of the symplectic measure relative to the Lebesgue 
measure,
\begin{equation*}
d \! \left( E \right) 
= \frac{ \dd \left( H_{*} \Omega \right) }{ \dd E }
= \frac{ \int_{H^{-1} (E) } \iota^{*}_{E} \left( v \, \lrcorner \, \Omega \right) }{ \dd H \! \left( v \right) },
\end{equation*}
is known as the density of states in the statistical mechanics literature \citep{Kardar:2007}.


\section{Hamiltonian Monte Carlo} \label{sec:hamiltonian_monte_carlo}

Although Hamiltonian systems feature exactly the kind of measure-preserving flow that could
generate an efficient Markov transition, there is no canonical way of endowing a given
probability space with a symplectic form, let alone a Hamiltonian.  In order take advantage
of Hamiltonian flow we need to consider not the sample space of interest but rather its
\textit{cotangent bundle}.

In this section we develop the formal construction of Hamiltonian Monte Carlo 
and identify how the theory informs practical considerations in both implementation
and optimal tuning.  Lastly we reconsider a few existing Hamiltonian Monte Carlo
implementations with this theory in mind.


\subsection{Formal Construction}

The key to Hamiltonian Monte Carlo is that the cotangent bundle of the sample
space, $T^{*} Q$, is endowed with both a canonical fiber bundle structure, 
$\pi : T^{*} Q \rightarrow Q$, and a canonical symplectic form.  If we can lift
the target distribution onto the cotangent bundle then we can construct an
appropriate Hamiltonian system and leverage its Hamiltonian flow to generate
a powerful Markov kernel.  When the sample space is also endowed with a 
Riemannian metric this construction becomes particularly straightforward.

\subsubsection{Constructing a Hamiltonian System} \label{sec:constructing_hmc}

By Corollary \ref{cor:joint_from_marginal}, the target distribution, $\varpi$, is lifted 
onto the cotangent bundle with the choice of a smooth disintegration, 
$\xi \in \Xi \! \left( \pi : T^{*} Q \rightarrow Q \right)$,
\begin{equation*}
\varpi_{H} = \pi^{*} \varpi \wedge \xi.
\end{equation*}
Because $\varpi_{H}$ is a smooth probability measure it must be of the form of a 
canonical distribution for some Hamiltonian $H : T^{*} Q \rightarrow \mathbb{R}$,
\begin{equation*}
\varpi_{H} = e^{-H} \, \Omega,
\end{equation*}
with $\beta$ taken to be unity without loss of generality.  In other words, the choice of a 
disintegration defines not only a lift onto the cotangent bundle but also a Hamiltonian 
system (Figure \ref{fig:cylinder_hmc}) with the Hamiltonian
\begin{equation*}
H = - \log \frac{ \dd \left( \pi^{*} \varpi \wedge \xi \right) }{ \dd \Omega }.
\end{equation*}

\begin{figure*}
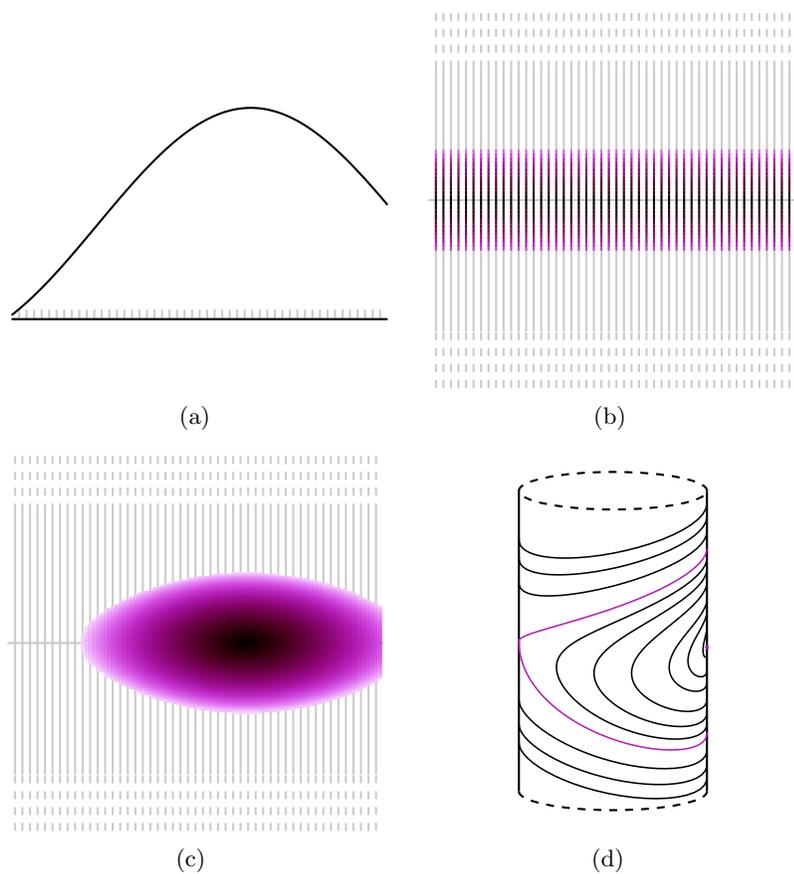

\centering
\subfigure[]{
\begin{tikzpicture}[scale=0.25, thick]
  
  \fill[color=white] (-10, -10) rectangle (10, 12);
  
  \foreach \x in {2,4,...,98} {
    \draw[color=gray80] ({\x / 5 - 10}, -6) -- (({\x / 5 - 10}, -5.5);
  }
  
  \draw[] (-10, -6) -- (10, -6);
  
  \node[] at (0, 0) {\includegraphics[width=5cm]{pushforward_density.eps}};
  
\end{tikzpicture}
}
\subfigure[]{
\begin{tikzpicture}[scale=0.25, thick]
  
  \foreach \x in {2,4,...,98} {
    \draw[style=dashed, color=gray80] ({\x / 5 - 10}, -10) -- ({\x / 5 - 10}, -7);
    \draw[color=gray80] ({\x / 5 - 10}, -7) -- ({\x / 5 - 10}, 7);
    \draw[style=dashed, color=gray80] ({\x / 5 - 10}, 7) -- ({\x / 5 - 10}, 10);
  }
  
  \draw[color=gray80] (-10, 0) -- (10, -0);
  
  \node[] at (-0.1, 0) {\includegraphics[width=4.95cm]{cylinder_density_disint.eps}};
  
\end{tikzpicture}
}
\subfigure[]{
\begin{tikzpicture}[scale=0.25, thick]
  
  \foreach \x in {2,4,...,98} {
    \draw[style=dashed, color=gray80] ({\x / 5 - 10}, -10) -- ({\x / 5 - 10}, -7);
    \draw[color=gray80] ({\x / 5 - 10}, -7) -- ({\x / 5 - 10}, 7);
    \draw[style=dashed, color=gray80] ({\x / 5 - 10}, 7) -- ({\x / 5 - 10}, 10);
  }
  
  \draw[color=gray80] (-10, 0) -- (10, -0);
  
  \node[] at (0,0) {\includegraphics[width=4.95cm]{cylinder_density_flat.eps}};
  
\end{tikzpicture}
}
\subfigure[]{
\begin{tikzpicture}[scale=0.25, thick]
  
  \draw[color=white] (-10.5, 0) -- (10.5, 0);
  
  \draw[] (-5, -8) -- (-5, 8);
  \draw[] (5, -8) -- (5, 8);
  \draw[style=dashed] (-5, -8) arc (180:360:5 and 1);
  \draw[style=dashed] (0, 8) ellipse (5 and 1);
  
  \node[] at (0,0) {\includegraphics[width=2.5cm]{cylinder_hamiltonian.eps}};
  \fill[color=highlight] (5, -0.33) circle (3pt);
  
\end{tikzpicture}
}
\caption{(a) Hamiltonian Monte Carlo begins with a target measure on the base space,
for example $Q = \mathbb{S}^{1}$.  (b) The choice of a disintegration on the cotangent
bundle, $T^{*} Q = \mathbb{S}^{1} \times \mathbb{R}$, defines (c) a joint measure on
the cotangent bundle which immediately defines (d) a Hamiltonian system given the
canonical symplectic structure.  The Hamiltonian flow of this system is then used to
construct an efficient Markov transition.}
\label{fig:cylinder_hmc}
\end{figure*}

Although this construction is global, it is often more conveniently implemented in local
coordinates.  Consider first a local neighborhood of the sample space,
$\mathcal{U}_{\alpha} \subset Q$, in which the target distribution decomposes as
\begin{equation*}
\varpi = e^{-V} \dd q^{1} \wedge \ldots \wedge \dd q^{n}.
\end{equation*}
Here $e^{-V}$ is the Radon--Nikodym derivative of the target measure with respect
to the pullback of the Lebesgue measure on the image of the local chart.  Following
the natural analogy to the physical application of Hamiltonian systems, we will refer 
to $V$ as the \textit{potential energy}.  

In the corresponding neighborhood of the cotangent bundle, 
$\pi^{-1} \! \left( \mathcal{U}_{\alpha} \right) \subset T^{*} Q$, the smooth disintegration, 
$\xi$, similarly decomposes into,
\begin{equation*}
\xi = e^{-T} \dd p_{1} \wedge \ldots \wedge \dd p_{n}
+ \text{horizontal} \; n\text{-forms}.
\end{equation*}
When $\xi$ is pulled back onto a fiber all of the horizontal $n$-forms vanish and $e^{-T}$ 
can be considered the Radon--Nikodym derivative of the disintegration restricted to a 
fiber with respect to the Lebesgue measure on that fiber.  Appealing to the physics 
conventions once again, we denote $T$ as the \textit{kinetic energy}.  

Locally the lift onto the cotangent bundle becomes
\begin{align*}
\varpi_{H} 
&= \pi^{*} \varpi \wedge \varpi_{q}
\\
&=
e^{- \left( T + V \right) } 
\dd q^{1} \wedge \ldots \wedge \dd q^{n} \wedge \dd p_{1} \wedge \ldots \wedge \dd p_{n}
\\
&=
e^{-H} \Omega,
\end{align*}
with the Hamiltonian
\begin{equation*}
H = - \log \frac{ \dd \varpi_{H} }{ \dd \Omega } = T + V,
\end{equation*}
taking a form familiar from classical mechanics~\citep{JoseEtAl:1998}.

A particular danger of the local perspective is that neither the potential energy, $V$, or the
kinetic energy, $T$, are proper scalar functions.  Both depend on the choice of chart and
introduce a log determinant of the Jacobian when transitioning between charts with
coordinates $q$ and $q'$,
\begin{align*}
V &\rightarrow V + \log \left| \frac{ \partial q }{ \partial q' } \right| \\
T&\rightarrow T - \log \left| \frac{ \partial q }{ \partial q' } \right|;
\end{align*}
only when $V$ and $T$ are summed do the chart-dependent terms cancel to give a 
scalar Hamiltonian.  When these local terms are used to implement 
Hamiltonian Monte Carlo care must be taken to avoid any sensitivity to the arbitrary
choice of chart, which usually manifests as pathological behavior in the algorithm.

\subsubsection{Constructing a Markov Transition}

Once on the cotangent bundle the Hamiltonian flow generates isomorphisms that preserve 
$\varpi_{H}$, but in order to define an isomorphism on the sample space we first need to map 
to the cotangent bundle and back.  

If $q$ were drawn from the target measure then we could generate an exact sample from 
$\varpi_{H}$ by sampling directly from the measure on the corresponding fiber,
\begin{equation*}
p \sim \iota_{q}^{*} \xi = e^{-T} \dd p_{1} \wedge \ldots \wedge \dd p_{n}.
\end{equation*}
In other words, sampling along the fiber defines a lift from $\varpi$ to $\varpi_{H}$,
\begin{align*}
\lambda :& \; Q \rightarrow T^{*} Q 
\\
& \; q \mapsto \left( p, q \right), p \sim \iota_{q}^{*} \xi.
\end{align*}
In order to return to the sample space we use the canonical projection, which by
construction maps $\varpi_{H}$ back into its pushforward, $\varpi$.

Together we have a random lift, 
\begin{align*}
\lambda : Q & \rightarrow T^{*} Q
\\
\lambda_{*} \varpi &= \varpi_{H},
\end{align*}
the Hamiltonian flow,
\begin{align*}
\phi^{H}_{t} : T^{*} Q & \rightarrow T^{*} Q
\\
\left( \phi^{H}_{t} \right)_{*} \varpi_{H} &= \varpi_{H},
\end{align*}
and finally the projection,
\begin{align*}
\pi : T^{*} Q & \rightarrow Q
\\
\pi_{*} \varpi_{H} &= \varpi.
\end{align*}
Composing the maps together,
\begin{equation*}
\phi_{\mathrm{HMC}} = \pi \circ \phi^{H}_{t} \circ \lambda
\end{equation*}
yields exactly the desired measure-preserving isomorphism,
\begin{align*}
\phi_{\mathrm{HMC}} : Q & \rightarrow Q
\\
\left( \phi_{\mathrm{HMC}} \right)_{*} \varpi &= \varpi,
\end{align*}
for which we have been looking.

Finally, the measure on $\lambda$, and possibly a measure on the integration time, $t$, 
specifies a measure on $\phi_{\mathrm{HMC}}$ from which we can define a Hamiltonian 
Monte Carlo transition via \eqref{eqn:kernel_from_isomorphisms}.

\subsubsection{Constructing an Explicit Disintegration}

The only obstacle with implementing Hamiltonian Monte Carlo as constructed is that
the disintegration is left completely unspecified.  Outside of needing to sample from
$\iota^{*}_{q} \xi$ there is little motivation for an explicit choice.

This choice of a smooth disintegration, however, is greatly facilitated by endowing the base 
manifold with a Riemannian metric, $g$, which provides two canonical objects from which 
we can construct a kinetic energy and, consequently, a disintegration.  Denoting 
$\tilde{p} \! \left( z \right)$ as the element of $T^{*}_{\pi (z)} Q$ identified by $z \in T^{*} Q$,
the metric immediately defines a scalar function, 
$g^{-1}\! \left( \tilde{p} \! \left( z \right), \tilde{p} \! \left( z \right) \right)$ and a density, 
$\left| g \! \left( \pi \! \left( z \right) \right) \right|$.  From the perspective of geometry
measure theory this latter term is just the Hausdorff density; in molecular dynamics
it is known as the Fixman potential \citep{Fixman:1978}.

Noting that the quadratic function is a scalar function where as the log density transforms
like the kinetic energy, an immediate candidate for the kinetic energy is given by simply
summing the two together,
\begin{equation*}
T \! \left( z \right) = 
\frac{1}{2} g^{-1}\! \left( \tilde{p} \! \left( z \right), \tilde{p} \! \left( z \right) \right) 
+ \frac{1}{2} \log \left| g \! \left( \pi \! \left( z \right) \right) \right| + \mathrm{const},
\end{equation*}
or in coordinates,
\begin{equation*}
T \! \left( p, q \right) = 
\frac{1}{2} \sum_{i, j = 1}^{n} p_{i} p_{j} \left( g^{-1} \! \left( q \right) \right)^{ij} 
+ \frac{1}{2} \log \left| g \! \left( q \right) \right| + \mathrm{const},
\end{equation*}
which defines a Gaussian measure on the fibers,
\begin{equation*}
\iota^{*}_{q} \xi = e^{-T} \dd p_{1} \wedge \ldots \wedge \dd p_{n} = \mathcal{N} \! \left( 0, g \right).
\end{equation*}
Using these same two ingredients we could also construct, for example, 
a multivariate Student's $t$ measure,
\begin{align*}
T \! \left( z \right) &= 
\frac{\nu + n}{2} \log \left( 
1+ \frac{1}{\nu} g^{-1}\! \left( \tilde{p} \! \left( z \right), \tilde{p} \! \left( z \right) \right) 
\right) + \frac{1}{2} \log \left| g \! \left( \pi \! \left( z \right) \right) \right|  + \mathrm{const} \\
\iota^{*}_{q} \xi &= t_{\nu} \! \left(0, g \right),
\end{align*}
or any distribution whose sufficient statistic is the Mahalanobis distance.

When $g$ is taken to be Euclidean the resulting algorithm is exactly the Hamiltonian
Monte Carlo implementation that has dominated both the literature and applications 
to date; we refer to this implementation as \textit{Euclidean Hamiltonian Monte Carlo}.
The more general case, where $g$ varies with position, is exactly
\textit{Riemannian Hamiltonian Monte Carlo}~\citep{GirolamiEtAl:2011} which has shown
promising success when the metric is used to correct for the nonlinearities of the target 
distribution.  In both cases, the natural geometric motivation for the choice of disintegration 
helps to explains why the resulting algorithms have proven so successful in practice.


\subsection{Practical Implementation} \label{sec:practical_implementation}

Ultimately the Hamiltonian Monte Carlo transition constructed above is a only a 
mathematical abstraction until we are able to simulate the Hamiltonian flow by
solving a system of highly-nonlinear, coupled ordinary differential equations.  At
this stage the algorithm is vulnerable to a host of pathologies and we have to
heed the theory carefully.

The numerical solution of Hamiltonian flow is a well-researched subject and many 
efficient integrators are available.  Of particular importance are symplectic integrators 
which leverage the underlying symplectic geometry to exactly preserve the symplectic 
measure with only a small error in the Hamiltonian~\citep{LeimkuhlerEtAl:2004, HairerEtAl:2006}.
Because they preserve the symplectic measure exactly, these integrators are highly
accurate even over long integration times.

Formally, symplectic integrators approximate the Hamiltonian flow by composing the flows 
generated from individual terms in the Hamiltonian.   For example, one second-order symplectic 
integrator approximating the flow from the Hamiltonian $H = H_{1} + H_{2}$ is given by
\begin{equation*}
\phi^{H}_{\delta t} =
\phi^{H_{1}}_{\delta t / 2} \circ 
\phi^{H_{2}}_{\delta t} \circ 
\phi^{H_{1}}_{\delta t / 2} + \mathcal{O} \! \left( \delta t ^{2} \right).
\end{equation*}
The choice of each component, $H_{i}$, and the integration of the resulting flow requires particular
care.  If the flows are not solved exactly then the resulting integrator no longer preserves
the symplectic measure and the accuracy plummets.  Moreover, each component must
be a scalar function on the cotangent bundle: although one might be tempted to take 
$H_{1} = V$ and $H_{2} = T$, for example, this would not yield a symplectic integrator 
as $V$ and $T$ are not proper scalar functions as discussed in Section \ref{sec:constructing_hmc}.
When using a Gaussian kinetic energy as described above, a proper decomposition is given by
\begin{align*}
H_{1} &= 
\frac{1}{2} \log \left| g \! \left( \pi \! \left( z \right) \right) \right| 
+ V \! \left( \pi \! \left( z \right) \right) \\
H_{2} &= 
\frac{1}{2} g^{-1}\! \left( \tilde{p} \! \left( z \right), \tilde{p} \! \left( z \right) \right).
\end{align*}

Although symplectic integrators introduce only small and well-understood errors, those errors 
will ultimately bias the resulting Markov chain.  In order to remove this bias we can consider 
the Hamiltonian flow not as a transition but rather as a Metropolis proposal on the cotangent 
bundle and let the acceptance procedure cancel any numerical bias.  Because they remain
accurate even for high-dimensional systems, the use of a symplectic integrator here is
crucial lest the Metropolis acceptance probability fall towards zero.

The only complication with a Metropolis strategy is that the numerical flow, $\Phi^{H}_{\epsilon, t}$, 
must be reversible in order to maintain detailed balance.  This can be accomplished by making 
the measure on the integration time symmetric about $0$, or by composing the flow 
with any operator, $R$, satisfying
\begin{equation*}
\Phi^{H}_{\epsilon, t} \circ R \circ \Phi^{H}_{\epsilon, t} = \mathrm{Id}_{T*Q}.
\end{equation*}
For all of the kinetic energies considered above, this is readily accomplished with a 
parity inversion given in canonical coordinates by
\begin{equation*}
R \! \left(q, p \right) = \left(q, -p \right).
\end{equation*}
In either case the acceptance probability reduces to
\begin{align*}
a \! \left( z, R \circ \Phi^{H}_{\epsilon, t} z \right)
&= 
\min \! \left[ 1, 
\exp \! \left( H \! \left( R \circ \Phi^{H}_{\epsilon, t} z \right) - H \! \left( z \right) \right) 
\right],
\end{align*}
where $\Phi^{H}_{\epsilon, t}$ is the symplectic integrator with step size, $\epsilon$,
and $z \in T^{*} Q$.


\subsection{Tuning}

Although the selection of a disintegration and a symplectic integrator formally define a
full implementation of Hamiltonian Monte Carlo, there are still free parameters left 
unspecified to which the performance of the implementation will be highly sensitive.  
In particular, we must set the integration time of the flow, the Riemannian metric, and 
symplectic integrator step size.  All of the machinery developed in our theoretical
construction proves essential here, as well.

The Hamiltonian flow generated from a single point may explore the entirety 
of the corresponding level set or be restricted to a smaller submanifold of the
level set, but in either case the trajectory nearly closes in some possibly-long 
but finite recurrence time, $\tau_{H^{-1} (E)}$~\citep{Petersen:1989, Zaslavsky:2005}.  
Taking
\begin{equation*}
t \sim U \! \left( 0, \tau_{H^{-1}(E)} \right),
\end{equation*}
would avoid redundant exploration but unfortunately the recurrence time for a given 
level set is rarely calculable in practice and we must instead resort to approximations.  
When using a Riemannian geometry, for example, we can appeal to the No-U-Turn 
sampler~\citep{HoffmanEtAl:2014, Betancourt:2013a} which has proven an empirical
success.

When using such a Riemannian geometry, however, we must address the fact that
the choice of metric is itself a free parameter.  One possible criterion to consider
is the interaction of the geometry with the symplectic integrator -- in the case of
a Gaussian kinetic energy, integrators are locally optimized when the metric 
approximates the Hessian of the potential energy, essentially canceling the local
nonlinearities of the target distribution.  This motivates using the global covariance 
of the target distribution for Euclidean Hamiltonian Monte Carlo and the SoftAbs 
metric~\citep{Betancourt:2013b} for Riemannian Hamiltonian Monte Carlo; for further
discussion see \cite{LivingstoneEtAl:2014}.  Although such choices work well in 
practice, more formal ergodicity considerations are required to define a more rigorous 
optimality condition.

Lastly we must consider the step size, $\epsilon$, of the symplectic integrator.  
As the step size is made smaller the integrator will become more accurate but also
more expensive -- larger step sizes yield cheaper integrators but at the cost of 
more Metropolis rejections.  When the target distribution decomposes into a product
of many independent and identically distributed measures, the optimal compromise 
between these extremes can be computed directly~\citep{BeskosEtAl:2013}.  More 
general constraints on the optimal step size, however, requires a deeper understanding 
of the interaction between the geometry of the exact Hamiltonian flow and that of the 
symplectic integrator developed in backwards error 
analysis~\citep{LeimkuhlerEtAl:2004, HairerEtAl:2006, IzaguirreEtAl:2004}.  In
particular, the microcanonical distribution constructed in Section \ref{sec:hamiltonian_systems}
plays a crucial role.


\subsection{Retrospective Analysis of Existing Work}

In addition to providing a framework for developing robust Hamiltonian Monte Carlo 
methodologies, the formal theory also provides insight into the performance of
recently published implementations.

We can, for example, now develop of deeper understanding of the poor scaling of the 
explicit Lagrangian Dynamical Monte Carlo algorithm~\citep{LanEtAl:2012}.  
Here the authors were concerned with the computation burden inherent to the implicit 
symplectic integrators necessary for Riemannian Hamiltonian Monte Carlo, and introduced 
an approximate integrator that sacrificed exact symplecticness for explicit updates.  As 
we saw in Section \ref{sec:practical_implementation}, however, exact symplecticness is
critical to maintaining the exploratory power of an approximate Hamiltonian flow,
especially as the dimension of the target distribution increases and numerical
errors amplify.  Indeed, the empirical results in the paper show that the performance of
the approximate integrator suffers with increasing dimension of the target distribution.
The formal theory enables an understanding of the compromises, and corresponding
vulnerabilities, of such approximations.

Moreover, the import of the Hamiltonian flow elevates the integration time as a 
fundamental parameter, with the integrator step size accompanying the use
of an approximate flow.  A common error in empirical optimizations is to 
reparameterize the integration time and step size as the number of integrator steps, 
which can obfuscate the optimal settings.  For example, \cite{WangEtAl:2013} use
Bayesian optimization methods to derive an adaptive Hamiltonian Monte Carlo
implementation, but they optimize the integrator step size and the number
of integrator steps over only a narrow range of values.  This leads not only to a
narrow range of short integration times that limits the efficacy of the Hamiltonian flow,
but also a \textit{step size-dependent} range of integration times that skew the 
optimization values.  Empirical optimizations of Hamiltonian Monte Carlo are most 
productive when studying the fundamental parameters directly.

In general, care must be taken to not limit the range of integration times considered
lest the performance of the algorithm be misunderstood.  For example, restricting the 
Hamiltonian transitions to only a small fraction of the recurrence time forfeits the 
efficacy of the flow's coherent exploration.  Under this artificial limitation, partial 
momentum refreshment schemes~\citep{Horowitz:1991, SohlEtAl:2014}, which 
compensate for the premature termination of the flow by correlating adjacent 
transitions, do demonstrate some empirical success.  As the restriction is withdrawn
and the integration times expand towards the recurrence time, however, the success 
of such schemes fade.  Ultimately, removing such limitations in the first place
results in more effective transitions.


\section{Future Directions}

By appealing to the geometry of Hamiltonian flow we have developed a formal,
foundational construction of Hamiltonian Monte Carlo, motivating various
implementation details and identifying the properties critical for a high
performance algorithm.  Indeed, these lessons have already proven critical in the
development of high-performance software like Stan~\citep{Stan:2014}.  Moving 
forward, the geometric framework not only admits further understanding and optimization 
of the algorithm but also suggests connections to other fields and motivates generalizations 
amenable to an even broader class of target distributions.

\subsection{Robust Implementations of Hamiltonian Monte Carlo}

Although we have constructed a theoretical framework in which we can pose rigorous
optimization criteria for the the integration time, Riemannian metric, and integrator step size,
there is much to be done in actually developing and then implementing those criteria.  
Understanding the ergodicity of Hamiltonian Monte Carlo is a critical step towards this goal, 
but a daunting technical challenge.  

Continued application of both the symplectic and Riemannian geometry underlying
implementations of the algorithm will be crucial to constructing a strong formal 
understanding of the ergodicity of Hamiltonian Monte Carlo and its consequences.
Initial applications of metric space methods~\citep{Ollivier:2009, JoulinEtAl:2010},
for example, have shown promise~\citep{HolmesEtAl:2014}, although many technical
obstacles, such as the limitations of geodesic completeness, remain.

\subsection{Relating Hamiltonian Monte Carlo to Other Fields}

The application of tools from differential geometry to statistical problems has rewarded
us with a high-performance and robust algorithm.  Continuing to synergize seemingly
disparate fields of applied mathematics may also prove fruitful in the future.

One evident association is to 
\textit{molecular dynamics}~\citep{Haile:1992, FrenkelEtAl:2001, MarxEtAl:2009}, which tackles 
expectations of chemical systems with natural Hamiltonian structures.  Although care must
be taken with the different construction and interpretation of the Hamiltonian from the statistical 
and the molecular dynamical perspectives, once a Hamiltonian system has been defined the 
resulting algorithms are identical.  Consequently molecular dynamics implementations may 
provide insight towards improving Hamiltonian Monte Carlo and vice versa.

Additionally, the composition of Hamiltonian flow with a random lift from the sample space onto 
its cotangent bundle can be considered a second-order stochastic 
process~\citep{BurrageEtAl:2007, Polettini:2013}, and the theory of these processes 
has the potential to be a powerful tool in understanding the ergodicity of the algorithm.

Similarly, the ergodicity of Hamiltonian systems has fueled a wealth of research into
dynamical systems in the past few decades~\citep{Petersen:1989, Zaslavsky:2005}.  
The deep geometric results emerging from this field complement those of the statistical 
theory of Markov chains and provide another perspective on the ultimate performance 
of Hamiltonian Monte Carlo.

The deterministic flow that powers Hamiltonian Monte Carlo is also reminiscent of various
strategies of removing the randomness in Monte Carlo 
estimation~\citep{Caflisch:1998, MurrayEtAl:2012, Neal:2012}.  The generality of the Hamiltonian 
construction may provide insight into the optimal compromise between random and deterministic 
algorithms.

Finally there is the possibility that the theory of measures on manifolds may be of
use to the statistical theory of smooth measures in general.  The application of
differential geometry to Frequentist methods that has consolidated into Information 
Geometry~\citep{AmariEtAl:2007} has certainly been a great success, and the use
Bayesian methods developed here suggests that geometry's domain of applicability 
may be even broader.  As demonstrated above, for example, the geometry of fiber bundles 
provides a natural setting for the study and implementation of conditional probability 
measures, generalizing the pioneering work of~\cite{Tjur:1980}.  

\subsection{Generalizing Hamiltonian Monte Carlo}

Although we have made extensive use of geometry in the construction of Hamiltonian 
Monte Carlo, we have not yet exhausted its utility towards Markov Chain Monte Carlo.  
In particular, further geometrical considerations suggest tools for targeting multimodal, 
trans-dimensional, infinite-dimensional, and possibly discrete distributions.

Like most Markov Chain Monte Carlo algorithms, Hamiltonian Monte Carlo
has trouble exploring the isolated concentrations of probability inherent to
multimodal target distributions.  Leveraging the geometry of contact manifolds,
however, admits not just transitions within a single canonical distribution but also 
transitions between different canonical distributions.  The resulting 
Adiabatic Monte Carlo provides a geometric parallel to simulated 
annealing and simulated tempering without being burdened by their common 
pathologies~\citep{Betancourt:2014}.

Trans-dimensional target distributions are another obstacle for Hamiltonian
Monte Carlo because of the discrete nature of the model space.  Differential
geometry may too prove fruitful here with the Poisson geometries that generalize 
symplectic geometry by allowing for a symplectic form whose rank need not be
constant~\citep{Weinstein:1983}.

Many of the properties of smooth manifolds critical to the construction of Hamiltonian
Monte Carlo do not immediately extend to the infinite-dimensional target distributions 
common to functional analysis, such as the study of partial differential 
equations~\citep{CotterEtAl:2013}.  Algorithms on infinite-dimensional spaces motivated 
by Hamiltonian Monte Carlo, however, have shown promise~\citep{BeskosEtAl:2011} 
and suggest that infinite-dimensional manifolds admit symplectic structures, or the 
appropriate generalizations thereof.

Finally there is the question of fully discrete spaces from which we cannot apply
the theory of smooth manifolds, let alone Hamiltonian systems.  Given that Hamiltonian 
flow can also be though of as an orbit of the symplectic group, however, there may 
be more general group-theoretic constructions of measure-preserving orbits that can 
be applied to discrete spaces.  


\section*{Acknowledgements}

We thank Tom LaGatta for thoughtful comments and discussion on disintegrations
and Chris Wendl for invaluable assistance with formal details of the geometric
constructions, but claim all errors as our own.  The preparation of this paper benefited 
substantially from the careful readings and recommendations of Saul Jacka, Pierre Jacob, 
Matt Johnson, Ioannis Kosmidis, Paul Marriott, Yvo Pokern, Sebastian Reich, and Daniel Roy.  
This work was motivated by the initial investigations in \cite{BetancourtEtAl:2011}.

Michael Betancourt is supported under EPSRC grant EP/J016934/1,  Simon Byrne is 
a EPSRC Postdoctoral Research Fellow under grant EP/K005723/1,  Samuel Livingstone 
is funded by a PhD scholarship from Xerox Research Center Europe, and Mark Girolami 
is an EPSRC Established Career Research Fellow under grant EP/J016934/1.


\appendix
\section{Proofs} \label{sec:proofs}

Here we collect the proofs of the Lemmas introduced in Section \ref{sec:measures_on_manifolds}.

\newtheorem*{lem:forms_as_measures}{Lemma \ref{lem:forms_as_measures}}
\begin{lem:forms_as_measures}

If $Q$ is a positively-oriented, smooth manifold then $\mathcal{M} \! \left( Q \right)$ 
is non-empty and its elements are $\sigma$-finite measures on $Q$.

\end{lem:forms_as_measures}

\begin{proof}

We begin by constructing a prototypical element of $\mathcal{M} \! \left( Q \right)$.
In a local chart $\left\{ \mathcal{U}_{\alpha}, \psi_{\alpha} \right\}$ we can construct a 
positive $\mu_{\alpha}$ as 
$\mu_{\alpha} = f_{\alpha} \, \dd q^{1} \wedge \ldots \wedge \dd q^{n}$ for any
$f_{\alpha} : \mathcal{U}_{\alpha} \rightarrow \mathbb{R}^{+}$.
Given the positive orientation of $Q$, the $\mu_{\alpha}$ are convex and we can define 
a global $\mu \in \mathcal{M} \! \left( Q \right)$ by employing a partition of unity
subordinate to the $\mathcal{U}_{\alpha}$,
\begin{equation*}
\mu = \sum_{\alpha} \rho_{\alpha} \mu_{\alpha}.
\end{equation*}

To show that any $\mu \in \mathcal{M} \! \left( Q \right)$ is a measure, consider
the integral of $\mu$ over any $A \in \mathcal{B} \! \left( Q \right)$. By construction
\begin{equation*}
\mu \! \left( A \right) = \int_{A} \mu > 0,
\end{equation*}
leaving us to show that $\mu \! \left( A \right)$ satisfies countable additivity and vanishes
when $A = \emptyset$.  We proceed by covering $A$ in charts and employing a partition 
of unity to give
\begin{align*}
\int_{A} \mu 
&= \sum_{\alpha} \int_{A \, \cap \, \mathcal{U}_{\alpha} } \rho_{\alpha} \, \mu_{\alpha} \\
&= \sum_{\alpha} \int_{A \, \cap \, \mathcal{U}_{\alpha} }
\rho_{\alpha} \, f_{\alpha} \, \dd q^{1} \wedge \ldots \wedge \dd q^{n}
\\
&= \sum_{\alpha} \int_{ \psi_{\alpha} \left( A \, \cap \, \mathcal{U}_{\alpha} \right) } 
\left( \rho_{\alpha} f_{\alpha} \circ \psi^{-1}_{\alpha} \right) \dd^{n} q,
\end{align*}
where $f_{\alpha}$ is defined as above and $\dd^{n} q$ is the Lebesgue measure on the
domain of the charts.

Now each domain of integration is in the $\sigma$-algebra of the sample space,
\begin{equation*}
A \, \cap \, \mathcal{U}_{\alpha} \in \mathcal{B} \! \left( Q \right),
\end{equation*}
and, because the charts are diffeomorphic and hence Lebesgue measurable functions, we must 
have
\begin{equation*}
\psi_{\alpha} \! \left( A \, \cap \, \mathcal{U}_{\alpha} \right) 
\in \mathcal{B} \! \left( \mathbb{R}^{n} \right).
\end{equation*}
Consequently the action of $\mu$ on $A$ decomposes into a countable number of Lebesgue 
integrals, and $\mu \! \left( A \right)$ immediately inherits countable additivity.

Moreover, $\psi_{\alpha} \left( \emptyset \, \cap \, \mathcal{U}_{\alpha} \right) =
\psi_{\alpha} \left( \emptyset \right) = \emptyset$ so that, by the same construction
as above,
\begin{align*}
\mu \! \left( \emptyset \right)
&=
\int_{\emptyset} \mu
\\
&= \sum_{\alpha} \int_{ \psi_{\alpha} \left( \emptyset \, \cap \, \mathcal{U}_{\alpha} \right) } 
\left( \rho_{\alpha} f_{\alpha} \circ \psi^{-1}_{\alpha} \right) \dd^{n} q
\\
&= \sum_{\alpha} \int_{ \emptyset } 
\left( \rho_{\alpha} f_{\alpha} \circ \psi^{-1}_{\alpha} \right) \dd^{n} q
\\
&= 0.
\end{align*}

Finally, because $Q$ is paracompact any $A \in \mathcal{B} \! \left( Q \right)$ admits a 
locally-finite refinement and, because any $\mu \in \mathcal{M} \! \left( Q \right)$ is 
smooth, the integral of $\mu$ over the elements of any such refinement are also finite.
Hence $\mu$ itself is $\sigma$-finite.

\end{proof}

\newtheorem*{lem:existence_of_fiber_volume_forms}{Lemma \ref{lem:existence_of_fiber_volume_forms}}
\begin{lem:existence_of_fiber_volume_forms}

The space $\Upsilon \! \left( \pi : Z \rightarrow Q \right)$ is convex and nonempty.

\end{lem:existence_of_fiber_volume_forms}

\begin{proof}

The convexity of $\Upsilon \! \left( \pi : Z \rightarrow Q \right)$ follows immediately from
the convexity of the positivity constraint and admits the construction of elements with a 
partition of unity.

In any neighborhood of a trivializing cover, $\left\{ \mathcal{U}_{\alpha} \right\}$, we have 
\begin{equation*}
\Upsilon \! \left( \pi^{-1} ( \mathcal{U}_{\alpha} ) \right) = \mathcal{M}^{+} \! \left( F \right)
\end{equation*}
which is nonempty by Corollary \ref{cor:forms_as_finite_measures}. Selecting some 
$\upsilon_{\alpha} \in \mathcal{M}^{+} \! \left( F \right)$ for each $\alpha$
and summing over each neighborhood gives
\begin{equation*}
\upsilon = \sum_{\alpha} \left( \rho_{\alpha} \circ \pi \right) \upsilon_{\alpha} \in \Upsilon \! \left( Z \right).
\end{equation*}
as desired.

\end{proof}

\newtheorem*{lem:fiber_forms_as_disintegrating_kernels}{Lemma \ref{lem:fiber_forms_as_disintegrating_kernels}}
\begin{lem:fiber_forms_as_disintegrating_kernels}

Any element of $\Upsilon \! \left( \pi : Z \rightarrow Q \right)$ defines a smooth measure,
\begin{align*}
\nu :
& \; Q \times \mathcal{B} \! \left( Z \right) \rightarrow \mathbb{R}^{+}
\\
& \; q, A \mapsto \int_{\iota_{q} \left( A \, \cap \, Z_{q} \right) } \iota_{q}^{*} \upsilon,
\end{align*}
concentrating on the fiber $Z_{q}$, 
\begin{equation*}
\nu \! \left( q, A \right) = 0, \,
\forall A \in \mathcal{B} \! \left( Z \right) | A \cap Z_{q} = 0.
\end{equation*}

\end{lem:fiber_forms_as_disintegrating_kernels}

\begin{proof}

By construction the measure of any $A \in \mathcal{B} \! \left( Z \right)$ is limited to its intersection
with the fiber $Z_{q}$, concentrating the measure onto the fiber.  Moreover, because the immersion
preserves the smoothness of $\upsilon$, $\iota^{*}_{q} \upsilon$ is smooth for all $q \in Q$ 
and the measure must be $\mathcal{B} \! \left( F \right)$-finite.  Consequently, the kernel is 
$\mathcal{B} \! \left( Z \right)$-finite.

\end{proof}

\newtheorem*{lem:measure_lift}{Lemma \ref{lem:measure_lift}}
\begin{lem:measure_lift}

Any element $\upsilon \in \Upsilon \! \left( \pi : Z \rightarrow Q \right)$ lifts any smooth measure
on the base space, $\mu_{Q} \in \mathcal{M} \! \left( Q \right)$, to a smooth measure on the total 
space by
\begin{equation*}
\mu_{Z} = \pi^{*} \mu_{Q} \wedge \upsilon \in \mathcal{M} \! \left( Z \right).
\end{equation*}

\end{lem:measure_lift}

\begin{proof}

Let $\left( X_{1} \! \left( q \right), \ldots, X_{n}  \! \left( q \right) \right)$ 
be a basis of 
$T_{q} Q, q \in Q,$ positively-oriented with respect to the $\mu_{Q}$,
\begin{equation*}
\mu_{Q} \left( X_{1}  \! \left( q \right), \ldots, X_{n}  \! \left( q \right) \right) > 0,
\end{equation*}
and $\left( Y_{1}  \! \left( p \right), \ldots, Y_{k}  \! \left( p \right) \right)$ a basis of 
$T_{p} Z_{q}, p \in Z_{q},$ positively-oriented with respect to the pull-back of $\upsilon$,
\begin{equation*}
\iota_{q}^{*} \upsilon \! \left( Y_{1} \! \left( p \right), \ldots, Y_{k} \! \left( p \right) \right) > 0.
\end{equation*}
Identifying $T_{p} Z_{q}$ as a subset of $T_{p} Z$, any horizontal lift of the $X_{i}  \! \left( q \right)$ 
to $\tilde{X}_{i}  \! \left( q \right) \in T_{p, q} Z$ yields a positively-oriented basis of the total space, 
$\left( \tilde{X}_{1}  \! \left( q \right), \ldots, \tilde{X}_{n} \! \left( q \right), 
Y_{1}  \! \left( p \right), \ldots, Y_{k}  \! \left( p \right)\right)$.
 
Now consider the contraction of this positively-oriented basis against 
$\mu_{Z} = \pi^{*} \mu_{Q} \wedge \omega$ for any $\omega \in \Omega^{k} \! \left( Z \right)$.  
Noting that, by construction, the $Y_{i}$ are vertical vectors and vanish when contracted against 
$\pi^{*} \mu_{Q}$, we must have
\begin{align*}
\mu_{Z} \! \left( \tilde{X}_{1}  \! \left( q \right), \ldots, \tilde{X}_{n}  \! \left( q \right), 
Y_{1}  \! \left( p \right), \ldots, Y_{k}  \! \left( p \right) \right)
&=
\pi^{*} \mu_{Q} \wedge \omega \! \left( \tilde{X}_{1}  \! \left( q \right), \ldots, \tilde{X}_{n}  \! \left( q \right), 
Y_{1}  \! \left( p \right), \ldots, Y_{k} \! \left( p \right) \right)
\\
&=
\pi^{*} \mu_{Q} \! \left( \tilde{X}_{1} \! \left( q \right), \ldots, \tilde{X}_{n} \! \left( q \right) \right)
\omega \! \left( Y_{1} \! \left( p \right), \ldots, Y_{k} \! \left( p \right) \right)
\\
&=
\mu_{Q} \! \left( X_{1} \! \left( q \right), \ldots, X_{n} \! \left( q \right)\right)
\omega \! \left( Y_{1} \! \left( p \right), \ldots, Y_{k} \! \left( p \right) \right)
\\
&> 0.
\end{align*}
Hence $\mu_{Z}$ is a volume form and belongs to $\mathcal{M} \! \left( Z \right)$.

Moreover, adding a horizontal $k$-form, $\eta$, to $\omega$ yields the same lift,
\begin{align*}
\mu_{Z}' \! \left( \tilde{X}_{1}  \! \left( q \right), \ldots, \tilde{X}_{n}  \! \left( q \right), 
Y_{1}  \! \left( p \right), \ldots, Y_{k}  \! \left( p \right) \right)
\\
& \hspace{-15mm} =
\pi^{*} \mu_{Q} \wedge \left( \omega + \eta \right) \! \left( \tilde{X}_{1}  \! \left( q \right), \ldots, \tilde{X}_{n}  \! \left( q \right), 
Y_{1}  \! \left( p \right), \ldots, Y_{k} \! \left( p \right) \right)
\\
& \hspace{-15mm} = \quad
\pi^{*} \mu_{Q} \! \left( \tilde{X}_{1} \! \left( q \right), \ldots, \tilde{X}_{n} \! \left( q \right) \right)
\omega \! \left( Y_{1} \! \left( p \right), \ldots, Y_{k} \! \left( p \right) \right)
\\
& \hspace{-15mm} \quad + 
\pi^{*} \mu_{Q} \! \left( \tilde{X}_{1} \! \left( q \right), \ldots, \tilde{X}_{n} \! \left( q \right) \right)
\eta \! \left( Y_{1} \! \left( p \right), \ldots, Y_{k} \! \left( p \right) \right)
\\
& \hspace{-15mm} =
\mu_{Q} \! \left( X_{1} \! \left( q \right), \ldots, X_{n} \! \left( q \right)\right)
\omega \! \left( Y_{1} \! \left( p \right), \ldots, Y_{k} \! \left( p \right) \right)
\\
& \hspace{-15mm} = \mu_{Z}.
\end{align*}
Consequently lifts are determined entirely by elements of the quotient space, 
$\upsilon \in \Upsilon \! \left( \pi : Z \rightarrow Q \right)$.

\end{proof}

\newtheorem*{lem:pushforward_measures}{Lemma \ref{lem:pushforward_measures}}
\begin{lem:pushforward_measures}

Let $\mu_{Z}$ be a smooth measure on the total space of a positively-oriented, smooth fiber bundle 
with $\mu_{Q}$ the corresponding pushforward measure with respect to the projection operator, 
$\mu_{Q} = \pi_{*} \mu_{Z}$.  If  $\mu_{Q}$ is a smooth measure then 
$\mu_{Z} = \pi^{*} \mu_{Q} \wedge \xi$ for a unique element of $\xi \in \Xi \! \left( \pi : Z \rightarrow Q \right)$.

\end{lem:pushforward_measures}

\begin{proof}

If the pushforward measure, $\mu_{Q}$, is smooth then it must satisfy
\begin{equation*}
\int_{B} \mu_{Q} = \int_{\pi^{-1} \left( B \right) } \mu_{Z}.
\end{equation*}

Employing a trivializing cover over $\pi^{-1} \! \left( B \right)$, we can expand the integral over
the total space as
\begin{align*}
\int_{\pi^{-1} \! \left( B \right) } \mu_{Z}
&=
\sum_{\alpha}
\int_{ \pi^{-1} \left( B \right) \, \cap \, \left( \mathcal{U}_{\alpha} \times F \right) } 
\rho_{\alpha} \, \mu_{Z}
\\
&=
\sum_{\alpha} \int_{ \left( B \, \cap \, \mathcal{U}_{\alpha} \right) \times F }
\rho_{\alpha} \, \mu_{Z}
\end{align*}

Following Theorem \ref{thm:existence_of_disintegrations} there is a unique 
$\upsilon \in \Upsilon \! \left( \pi : Z \rightarrow Q \right)$ such that 
$\mu_{Z} = \pi^{*} \mu_{Q} \wedge \upsilon$ and the integral becomes
\begin{align*}
\int_{\pi^{-1} \! \left( B \right) } \mu_{Z}
&=
\sum_{\alpha} \int_{ \left( B \, \cap \, \mathcal{U}_{\alpha} \right) \times F }
\rho_{\alpha} \, \mu_{Z}
\\
&=
\sum_{\alpha} \int_{ \left( B \, \cap \, \mathcal{U}_{\alpha} \right) \times F }
\rho_{\alpha} \, \pi^{*} \mu_{Q} \wedge \upsilon
\\
&=
\sum_{\alpha} \int_{ \left( B \, \cap \, \mathcal{U}_{\alpha} \right) }
\rho_{\alpha} \left[ \int_{F} \iota^{*}_{q} \upsilon \right] \mu_{Q}
\\
&=
\int_{ B } \rho_{\alpha}
\left[ \int_{ Z_{q} } \iota^{*}_{q} \upsilon \right] \mu_{Q}.
\end{align*}
Because $\mu_{Q}$ is $\sigma$-finite and $\int_{ Z_{q} } \iota^{*}_{q} \upsilon$ is finite,
the pushforward condition is satisfied if and only if
\begin{equation*}
\int_{ Z_{q} } \iota^{*}_{q} \upsilon = 1, \forall q \in Q,
\end{equation*}
which is satisfied if and only if 
$\upsilon \in \Xi \! \left( \pi : Z \rightarrow Q \right) \subset \Upsilon \! \left( \pi : Z \rightarrow Q \right)$.

Consequently there exists a unique $\xi \in \Xi \! \left (\pi : Z \rightarrow Q \right)$ that lifts the 
pushforward measure of $\mu_{Z}$ back to $\mu_{Z}$.

\end{proof}

\newtheorem*{lem:canonical_disintegration}{Lemma \ref{lem:canonical_disintegration}}
\begin{lem:canonical_disintegration}

Let $\mu_{Z}$ be a smooth measure on the total space of a positively-oriented, smooth fiber bundle
whose pushforward measure with respect to the projection operator is smooth, with $U \subset Q$
any neighborhood of the base space that supports a local frame.  Within $\pi^{-1} \! \left( U \right)$, 
the element $\xi \in \Xi \! \left( \pi : Z \rightarrow Q \right)$
\begin{equation*}
\xi = 
\frac{ \left( \tilde{X}_{1}, \ldots, \tilde{X}_{n} \right) \, \lrcorner \, \mu_{Z} }
{ \mu_{Q} \! \left( X_{1}, \ldots, X_{n} \right) },
\end{equation*}
defines the regular conditional probability measure of $\mu_{Z}$ with respect to the projection operator,
where $\left( X_{1}, \ldots, X_{n} \right)$ is any positively-oriented frame in $U$ satisfying
\begin{equation*}
\mu_{Q} \! \left( X_{1}, \ldots, X_{n} \right) < \infty, \, \forall q \in U
\end{equation*}
and $\left( \tilde{X}_{1}, \ldots, \tilde{X}_{n} \right)$ is any corresponding horizontal lift.

\end{lem:canonical_disintegration}

\begin{proof}

Consider any positively-oriented frame on the base space, $\left( X_{1}, \ldots, X_{n} \right)$, along with
any choice of horizontal lift, $\left( \tilde{X}_{1}, \ldots, \tilde{X}_{n} \right)$, and an ordered $k$-tuple
of vector fields on the total space, $\left( Y_{1}, \ldots, Y_{k} \right)$, that restricts to a positively-ordered 
frame in some neighborhood of the fibers, $V \subset \pi^{-1} \! \left( U \right)$.  Because the 
fiber bundle is oriented, the horizontal lift and the ordered $k$-tuple define a positively-ordered frame
in $V$,
\begin{equation*}
\left( W_{1}, \ldots, W_{n + k} \right) = \left( \tilde{X}_{1}, \ldots, \tilde{X}_{n}, Y_{1}, \ldots, Y_{k} \right).
\end{equation*}

If the pushforward measure, $\mu_{Q}$, is smooth then from Lemma \ref{lem:pushforward_measures} 
we have $\mu_{Z} =  \pi^{*} \mu_{Q} \wedge \xi$ for some unique $\xi \in \Xi \! \left( \pi : Z \rightarrow Q \right)$.
Contracting the frame onto these forms gives
\begin{align*}
\mu_{Z} \! \left( W_{1}, \ldots, W_{n + k} \right)
&=
\left( \pi^{*} \mu_{Q} \wedge \xi \right) \! \left( W_{1}, \ldots, W_{n + k} \right)
\\
&=
\mu_{Q} \! \left( X_{1}, \ldots, X_{n} \right) \xi \! \left( Y_{1}, \ldots, Y_{k} \right).
\end{align*}
Given the positive orientations of the frames and the forms, each term is strictly positive and provided that 
$\mu_{Q} \! \left( X_{1}, \ldots, X_{n} \right)$ is finite for all $q \in U$ we can divide to give,
\begin{equation*}
\xi \! \left( Y_{1}, \ldots, Y_{k} \right) =
\frac{ \mu_{Z} \! \left( W_{1}, \ldots, W_{n + k} \right) }
{ \mu_{Q} \left( X_{1}, \ldots, X_{n} \right) }.
\end{equation*}

Finally, because $\xi$ is invariant to the addition of horizontal $k$-forms this implies that within $U$
\begin{equation*}
\xi =
\frac{ \left( \tilde{X}_{1}, \ldots, \tilde{X}_{n} \right) \, \lrcorner \, \mu_{Z} }
{ \mu_{Q} \! \left( X_{1}, \ldots, X_{n} \right) }.
\end{equation*}

\end{proof}

\newtheorem*{lem:invariance_of_microcanonical}{Lemma \ref{lem:invariance_of_microcanonical}}
\begin{lem:invariance_of_microcanonical}

Let $\left( M, \Omega, H \right)$ be a Hamiltonian system with the finite and smooth canonical measure,
$\mu = e^{- \beta H}$.  The microcanonical distribution on the level set $H^{-1} (E)$,
\begin{equation*}
\varpi_{H^{-1} (E) } = 
\frac{ v \, \lrcorner \, \Omega }
{ \int_{H^{-1} (E) } \iota^{*}_{E} \left( v \, \lrcorner \, \Omega \right) },
\end{equation*}
is invariant to the corresponding Hamiltonian flow restricted to the level set, 
$\left. \phi^{H}_{t} \right|_{H^{-1} (E) }$.

\end{lem:invariance_of_microcanonical}

\begin{proof}

By construction the global flow preserves the canonical distribution,
\begin{align*}
\varpi 
&= 
\left( \phi^{H}_{t} \right)_{*} \varpi
\\
&= 
\left( \phi^{H}_{t} \right)_{*} \! \left( \varpi_{H^{-1} (E) } \wedge \varpi_{E} \right)
\\
&= 
\left( \left( \phi^{H}_{t} \right)_{*} \varpi_{H^{-1} (E) } \right) \wedge 
\left( \left( \phi^{H}_{t} \right)_{*} \varpi_{E} \right).
\end{align*}
Because the Hamiltonian is itself invariant to the flow we must have
\begin{equation*}
\left( \left( \phi^{H}_{t} \right)_{*} \varpi_{E} \right) = \varpi_{E}
\end{equation*}
and
\begin{equation*}
\varpi = 
\left( \left( \phi^{H}_{t} \right)_{*} \varpi_{H^{-1} (E) } \right) \wedge 
\varpi_{E}.
\end{equation*}

From Lemma \ref{lem:pushforward_measures}, however, the regular conditional probability 
measure in the decomposition must be unique, hence
\begin{align*}
\iota_{q}^{*} \varpi_{H^{-1} (E) }
&=
\iota_{q}^{*} \left( \left( \phi^{H}_{t} \right)_{*} \varpi_{H^{-1} (E) } \right)
\\
&=
\left( \phi^{H}_{t} \circ \iota_{q} \right)_{*} \varpi_{H^{-1} (E) } 
\\
&=
\left( \iota_{q} \circ \left. \phi^{H}_{t} \right|_{H^{-1} (E) } \right)_{*} \varpi_{H^{-1} (E) } 
\\
&=
\left( \left. \phi^{H}_{t} \right|_{H^{-1} (E) } \right)_{*} \iota_{q}^{*} \varpi_{H^{-1} (E) } ,
\end{align*}
as desired.

\end{proof}

\bibliography{geo_of_hmc}
\bibliographystyle{imsart-nameyear}

\end{document}